\documentclass[journal]{IEEEtran}
\usepackage{amsmath,amsfonts}
\usepackage{algorithmic}
\usepackage{array}
\usepackage{textcomp}
\usepackage{stfloats}
\usepackage{url}
\usepackage{verbatim}
\usepackage{cite}
\usepackage{sharina}
\usepackage{setspace}
\usepackage{amsthm}
\usepackage{amssymb,amsmath,latexsym}
\usepackage{graphicx,subfigure}
\usepackage{color,cite,times}
\usepackage{epstopdf}
\usepackage{multirow}
\usepackage{array}
\usepackage{booktabs}
\usepackage{hhline} 
\newtheorem{assumption}{Assumption}

\newtheorem{theorem}{\textbf{Theorem}}
\newtheorem{lemma}{\textbf{Lemma}}

\allowdisplaybreaks[4]
\usepackage[ruled,vlined,linesnumbered]{algorithm2e}
\usepackage{amsmath,amssymb}
\usepackage{mathtools} 

\begin{document}
	
	\title{
Event-Triggered Gossip for Distributed Learning}
\author{Zhiyuan Zhai, Xiaojun Yuan, \IEEEmembership{Fellow, IEEE}, 
Wei Ni, \IEEEmembership{Fellow, IEEE},\\
        Xin Wang, \IEEEmembership{Fellow, IEEE}, Rui Zhang, \IEEEmembership{Fellow, IEEE}, and  Geoffrey Ye Li, \IEEEmembership{Fellow, IEEE}}

	\maketitle
	
\begin{abstract}
While distributed learning offers a new learning paradigm for distributed network with no central coordination, it is constrained by communication bottleneck between nodes.
	We develop a new event-triggered gossip framework for distributed learning to reduce inter-node communication overhead. The framework introduces an adaptive communication control mechanism that enables each node to autonomously decide in a fully decentralized fashion when to exchange model information with its neighbors based on local model deviations. We analyze the  ergodic convergence of the proposed framework under noconvex objectives and interpret the convergence guarantees under different triggering conditions. Simulation results show that the proposed framework achieves substantially lower communication overhead  than the state-of-the-art distributed learning methods, reducing cumulative point-to-point transmissions by \textbf{71.61\%} with only a marginal performance loss, compared with the conventional full-communication baseline.
\end{abstract}


\section{Introduction}
This section provides an overview of distributed learning systems and the communication bottlenecks that limit their scalability in decentralized settings. We first discuss existing communication-reduction strategies and identify their limitations in adaptively controlling information exchange. Then, we present the motivation, contributions and structure of the event-triggered gossip framework proposed to address the limitations in this paper.
\subsection{Motivation and Challenges}
Machine learning has been widely used in  intelligent data analytics and decision-making in domains, such as autonomous driving,  Internet of Things (IoT), and smart healthcare.  
Centralized training usually aggregates raw data at a single location, and is sometimes infeasible due to communication bandwidth limitations, privacy constraints, and single-point-of-failure risks~\cite{tang2020commEfficientDL,Zhou2023InexactADMM,zhou2024largeScaleDL,10818523,zhang2025dark}.
Distributed learning, such as federated learning (FL), enables multiple devices, each holding its local dataset, to collaboratively train a shared model without relying on a central coordinator.  
Each node communicates only with its direct neighbors and exchanges model states or gradients to achieve consensus over time~\cite{lin2021quasiGlobalMom, xuanyu2023overviewDL,Qin2021FLWireless,Ye2022DecentralizedFL,Sha2025SparseDFL}.  
This architecture improves scalability, robustness, and data privacy, making it well suited for bandwidth- and energy-limited systems.

Despite these advantages, distributed learning often suffers from a severe \emph{communication bottleneck}.  
To maintain consensus among nodes, standard algorithms (e.g., gossip-based methods \cite{koloskova2019decentralized,song2022communication,zhai2025spectral}) require all devices to exchange information at every iteration, leading to massive point-to-point transmission and rapid exhaustion of bandwidth resources~\cite{hong2019dlion,Renggli2019sparcML}.  
On the other hand, excessive suppression of communication may lead to model drift across nodes and degrades learning accuracy~\cite{cao2021overviewCommEff, nabli2023acid}.  
This raises a central question:
	{How can we design communication-efficient distributed learning algorithms that preserve high model accuracy but significantly reduce inter-node transmissions?}

\subsection{Related Work}

 A variety of research efforts have  explored methods to reduce this communication burden. We  devide them into three major categories and discuss each of them.

\paragraph*{Model and Gradient Compression}  
One representative line of research reduces the size of exchanged updates through quantization or sparsification. 
For instance, ternary gradient (TernGrad)~\cite{wen2017terngrad} and  quantized stochastic gradient descent (QSGD)~\cite{alistarh2017qsgd} quantize gradients into low-bit representations while maintaining convergence guarantees. 
Later works extend these ideas by transmitting only the most significant components, as in Top-$k$ gradient sparsification~\cite{aji2017sparse} or memory-based error compensation methods~\cite{stich2018sparsified}. 
System-level designs, such as deep gradient compression (DGC)~\cite{lin2018dgc}, further combine momentum correction and local gradient clipping to achieve up to $600\times$ communication reduction without accuracy loss. 
These compression-based methods focus on minimizing the payload of each communication round.

\paragraph*{Periodic or Infrequent Communication}  
Another major direction reduces communication frequency by allowing devices to perform multiple local updates between synchronization rounds. 
The local stochastic gradient descent (SGD) framework~\cite{stich2019local} provides rigorous convergence analysis under delayed averaging, 
while\cite{haddadpour2019local} extends it with adaptive synchronization intervals. 
In FL, federated averaging (FedAvg)~\cite{mcmahan2017communication} adopts similar periodic aggregation and partial participation principles to trade off communication and computation. More recently,  hybrid FL algorithm (FedGiA)~\cite{Zhou2023FedGiA}  combines periodic communication with a hybrid gradient descent and inexact alternating direction method of
multipliers (ADMM) update to reduce communication rounds under mild convergence conditions.
A unified analysis of such approaches is provided  in cooperative SGD~\cite{wang2021cooperative}. 
These approaches reduce communication  by decreasing synchronization frequency, but may still suffer from model divergence under heterogeneous data distributions.

\paragraph*{Over-the-Air Aggregation}  
At the wireless physical layer, an emerging body of work exploits the superposition property of multiple-access channels to aggregate model updates “over the air.” 
Over-the-Air (OTA) computation~\cite{yang2018aircomp,zhu2019baa,amiri2020analog,sery2021otafl,10506083} enables simultaneous analog transmission of local gradients, such that the received waveform inherently represents their sum. 
Recent surveys~\cite{zhou2024ota,XIAO20251710} summarize advances in channel inversion, power control, and error compensation that make OTA a promising paradigm for low-latency, large-scale distributed learning. 
These techniques focus on improving the efficiency of the {physical-layer aggregation}.

Although these methods have substantially alleviated the communication burden in distributed learning, 
they  focus primarily on optimizing the {amount}, {frequency}, or {physical efficiency} of information exchange. 
Little attention has been paid to designing a communication control mechanism at the \emph{algorithmic level} that adaptively determines \emph{when} devices should exchange their local models. 
Such a mechanism represents a complementary perspective of optimization  to the above communication-reduction approaches, providing an alternative means to reduce inter-node communication while maintaining  learning accuracy.

\subsection{Contributions}

This paper develops a new event-triggered gossip framework for distributed learning, which significantly reduces redundant inter-node transmissions while maintaining model accuracy. 
The core idea is to introduce a communication control mechanism at the algorithmic level, allowing each device to autonomously decide {when} to exchange its model with neighbors based on local dynamics. 
We  prove that the proposed algorithm has the potential to achieve the same convergence rate as centralized SGD. 
Extensive experiments verify that our method can drastically reduce communication volume without sacrificing learning performance.

The key contributions of this paper are summarized as follows:
\begin{itemize}
	\item \textit{Event-triggered communication mechanism:}  
	We design a novel gossip-based distributed learning algorithm where each node  triggers communication according to its local model deviation. 
	This mechanism enables asynchronous and data-dependent information exchange without any global coordination, effectively suppressing redundant transmissions while preserving consensus.
	
	\item \textit{Unified convergence analysis:}  
	For the first time, we establish the  ergodic convergence bound for event-triggered distributed learning under nonconvex objectives. 
	Our analysis reveals that the convergence rate  depends jointly on the stepsize, the network spectral gap, and the triggering thresholds,  providing explicit insights into the trade-off between learning accuracy and communication cost.
	
\item \textit{Behavior under different triggering strategies:}  
We analyze the convergence of the proposed algorithm under three representative triggering policies, including constant zero threshold, fixed nonzero threshold, and gradually decaying threshold.  
It is revealed that under a decaying threshold policy, our algorithm achieves the same convergence rate as centralized SGD, i.e., $\mathcal{O}\!\big(T^{-1/2}\big)$.

	\item \textit{Comprehensive experimental validation:}  
	Extensive simulations on the MNIST \cite{lecun1998mnist} and Fashion-MNIST \cite{xiao2017fashion} datasets demonstrate that the proposed framework achieves comparable learning accuracy to the full communication scheme while substantially reducing the communication volume. 
	The algorithm reduces the cumulative point-to-point transmissions by \textbf{71.6\%} on Fashion-MNIST and \textbf{69.4\%} on MNIST, with less than 1\% accuracy degradation. 
\end{itemize}

The rest of this paper is organized as follows. 
Section~II presents the system model and event-triggered gossip formulation. 
Section~III provides the convergence analysis under different triggering strategies. 
Section~IV reports simulation results and performance comparisons. 
Finally, Section~V concludes the paper.

\textit{Notation:}
 $\mathbb{R}$  denotes the set of real numbers.  
We represent scalars by regular letters, vectors by bold lowercase letters, and matrices by bold uppercase letters.  
$(\cdot)^{\mathrm{T}}$ denotes  transpose.   
We denote the $i$-th largest eigenvalue of a matrix by $\lambda_i(\cdot)$,  
and the $(i,j)$-th entry of a matrix $\mathbf{X}$ by $\mathbf{X}_{ij}$.  
We denote the Euclidean norm for vectors/matrices by $\|\cdot\|$ or $\|\cdot\|_2$, and the Frobenius norm for matrices by $\|\cdot\|_F$.    
$\mathbb{E}[\cdot]$ denotes the expectation operator.  
We denote an all-one vector by $\mathbf{1}$ and the identity matrix by $\mathbf{I}$.  $\mathcal{O}(\cdot)$ denotes an upper bound up to a constant factor; 
$\Theta(\cdot)$ denotes a tight bound up to constant factors;
 $\mathbf e_i\in\mathbb{R}^n$ denotes the $i$-th canonical basis vector, whose $i$-th entry is $1$ and the rest are $0$.

\section{System Model}
\label{sec:system-model}


\begin{figure}[t]
	\centering
	\includegraphics[width=0.9\linewidth]{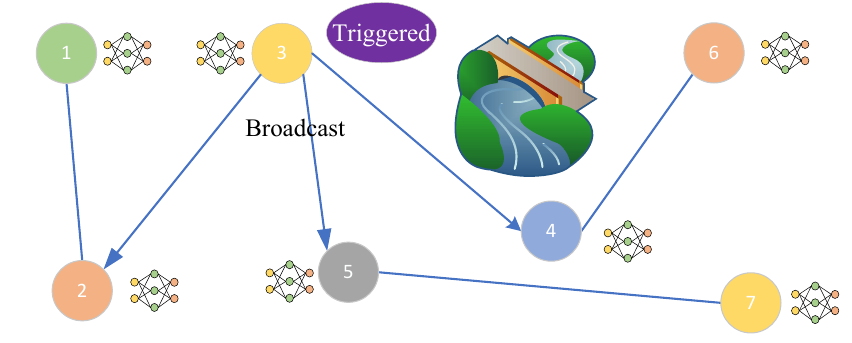}
	\caption{An example of event-triggered communication  model.}
	\label{sys_model}
\end{figure}



In this section, we present the  system model of distributed learning.
We start with the motivation for adopting a distributed training architecture,
followed by the corresponding network topology and optimization objective.

\subsection{Background}
In many modern applications, such as edge intelligence and IoT systems,
training data are  generated and stored across multiple devices
(e.g., mobile phones, sensors, or edge servers).
A straightforward solution is \emph{centralized learning},
where all devices upload their raw data to a central server,
which then trains a global model using the aggregated dataset.
This approach often suffers from several practical limitations:
(i) transmitting all raw data leads to  high communication overhead,
(ii) sharing raw data may violate privacy constraints,
and (iii) the central server becomes a single point of failure.

To overcome these issues, \emph{distributed learning} has emerged as an attractive alternative.
Instead of sending raw data to a central node,
each device keeps its local dataset and performs local training,
and only exchanges model parameters or gradients with neighboring devices.
In this way, the system collaboratively learns a shared model
with significantly reduced communication cost and improved privacy protection.

\subsection{Network and Learning Objective}

We model the considered system as a decentralized network consisting of $n$ computing nodes
$\mathcal{V}=\{1,\ldots,n\}$ connected by a static, undirected, and connected graph
$\mathcal{G}=(\mathcal{V},\mathcal{E})$, as illustrated in Fig.~\ref{sys_model}.
Here, $\mathcal{E}$ denotes the set of communication links.
If $(i,j)\in\mathcal{E}$, nodes $i$ and $j$ are said to be \emph{neighbors},
and $\mathcal{N}(i)$ denotes the neighbor set of node $i$.

Communication is restricted to \emph{one-hop neighbor exchanges only}:
each node can directly communicate only with its immediate neighbors in the graph.
Multi-hop relaying, centralized coordination, or global broadcasts are not considered.
This setting captures practical peer-to-peer or decentralized edge networks,
where only local connectivity is available.

Each node $i$ holds a local dataset drawn from a data distribution $\mathcal{D}_i$.
Based on its own data, node $i$  minimizes the local expected loss function,
\begin{equation}
	f_i(\mathbf{x})
	\triangleq
	\mathbb{E}_{\xi\sim\mathcal{D}_i}\!\big[\ell(\mathbf{x};\xi)\big],
	\qquad
	\mathbf{x}\in\mathbb{R}^{d},
\end{equation}
where $\mathbf{x}$ denotes the model parameter vector,
$\xi$ represents a random data sample at node $i$,
and $\ell(\mathbf{x};\xi)$ is the sample-wise loss
(e.g., the cross-entropy loss for classification or the squared loss for regression).

The goal of distributed learning is to collaboratively optimize a shared global model
by minimizing the average loss over all nodes in a peer-to-peer fashion, as given by
\begin{equation}
	\min_{\mathbf{x}\in\mathbb{R}^{d}}
	\; f(\mathbf{x})
	\triangleq
	\frac{1}{n}\sum_{i=1}^{n} f_i(\mathbf{x}).
\end{equation}
Since the data remains locally stored and only model information is exchanged,
the above problem must be solved through decentralized cooperation among neighboring nodes.

\section{Proposed Event-Triggered Framework}
In this section, we introduce the proposed event-triggered communication
framework and the corresponding local update rules.
We first explain the motivation for reducing communication in distributed learning,
and then present the event-triggered gossip mechanism, together with a compact
matrix-form description of the overall system dynamics.

\subsection{Motivation for Event-Triggered Communication}

In existing distributed learning algorithms,
communication typically follows a \emph{full-communication} strategy.
Specifically, at every iteration, each node exchanges its current model
(or gradient) with all of its neighbors and performs a consensus update.
Although this approach ensures fast information mixing across the network,
it requires communication at every round and over every link.

When the number of devices or the training horizon is large,
such frequent transmissions lead to substantial communication overhead,
which may quickly exhaust bandwidth and energy resources in practical edge networks.
In many scenarios, however, consecutive local models change only slightly,
making repeated transmissions largely redundant.

This observation motivates the following question:
\emph{Can nodes communicate only when necessary,
while  maintaining comparable learning accuracy and consensus performance?}
To answer this research question, we propose an \emph{event-triggered communication} mechanism.
Instead of communicating at every iteration,
each node autonomously decides whether to transmit its model
according to its local model evolution.
Communication is triggered only when the local model deviates
sufficiently from its last transmitted state,
thereby suppressing redundant exchanges and significantly reducing the overall communication cost.
\subsection{Event-Triggered Gossip}

We consider an iterative distributed learning process.
At  iteration $t$, every node maintains a local model,
computes a stochastic gradient based on its local data,
and exchanges model information with its neighbors to promote consensus.
The communication among nodes is governed by an event-triggered mechanism,
which adaptively determines whether a node should transmit its local model
based on model evolution.

Each iteration consists of the following three steps.

\textbf{Step 1: Communication (broadcast and receive).}

Let $\mathbf{x}_{i,t}$ denote the local model of node $i$ in the $t$-th iteration.
Each node $i$ maintains two  states in addition to $\mathbf{x}_{i,t}$:
\begin{itemize}
	\item \textbf{Broadcast snapshot} $\hat{\mathbf{x}}_{i,t}\in\mathbb{R}^{d}$: The last model of node $i$ broadcast to neighbors (initialized by $\hat{\mathbf{x}}_{i,0}=\mathbf{x}_{i,0}$).
	\item \textbf{Receive caches} $\tilde{\mathbf{x}}_{j\rightarrow i,t}\in\mathbb{R}^{d}$: At receiver $i$, the most recently received model from neighbor $j$ (initialized by $\tilde{\mathbf{x}}_{j\rightarrow i,-1}=\mathbf{x}_{j,0}$). We enforce $\tilde{\mathbf{x}}_{i\rightarrow i,t}\equiv \mathbf{x}_{i,t}$.
\end{itemize}
At the beginning of round $t$, each node $i$ computes the drift
\begin{equation}
	\mathbf{e}_i^{t}\triangleq \mathbf{x}_{i,t}-\hat{\mathbf{x}}_{i,t}.
\end{equation}
Given a nonnegative threshold $\tau_{t}\ge 0$, node $i$ \emph{triggers} communication if $\|\mathbf{e}_i^{t}\|\ge \tau_{t}$, immediately broadcasting $\mathbf{x}_{i,t}$ to all $j\in\mathcal{N}(i)$ and setting $\hat{\mathbf{x}}_{i,t}\!\leftarrow \mathbf{x}_{i,t}$.

After receiving the models, each node $i$ forms a convex combination of its own
and its neighbors’ models using the mixing weights
$\mathbf{W}\in\mathbb{R}^{n\times n}$:
\begin{equation}
	\mathbf{x}_{i,t,\mathrm{mix}} = \sum_{j=1}^{n} \mathbf{W}_{ji}\,\tilde{\mathbf{x}}_{j\rightarrow i,t},
	\label{eq:gossip}
\end{equation}
where $\tilde{\mathbf{x}}_{j\rightarrow i,t}$ denotes the cached model received
from node $j$, which is defined as
\begin{equation}
	\tilde{\mathbf{x}}_{j\rightarrow i,t} =
	\begin{cases}
		\mathbf{x}_{j,t}, & \text{if } \|\mathbf{e}_j^{t}\|\ge \tau_{t},\\
		\tilde{\mathbf{x}}_{j\rightarrow i,t-1}, & \text{otherwise (reuse cache).}
	\end{cases}
	\label{eq:cache}
\end{equation}
Moreover, $\mathbf{W}$ is a mixing matrix supported on the network graph
$\mathcal{G}$, satisfying
\begin{align}
    \mathbf{W}_{ji} > 0, &\quad \text{only if } j \in \mathcal{N}(i)\cup\{i\};\\
    \mathbf{W}_{ji} = 0, &\quad \text{otherwise}.
\end{align}

\textbf{Step 2: Local gradient computation.}

Each node $i$ computes a stochastic gradient on its local model by sampling $\xi_{i,t}$\footnote{
Since the stochastic gradient is evaluated at the current local model $\mathbf{x}_{i,t}$,
it does not depend on the newly received neighbor information.
Therefore, the gradient computation does not need to wait for the completion of the communication step,
and Steps 1 and 2 can be executed in parallel in practice.}:
\begin{equation}
	\mathbf{g}_{i,t} \triangleq g_i(\mathbf{x}_{i,t};\xi_{i,t}) 
	= \nabla_{\mathbf{x}}\ell(\mathbf{x}_{i,t};\xi_{i,t}).
\end{equation}

\textbf{Step 3: Model update.}

After that, each node $i$ updates its model using the mixing model $\mathbf{x}_{i,t,\mathrm{mix}}$ and the computed stochastic gradient:
\begin{equation}
	\mathbf{x}_{i,t+1} = \mathbf{x}_{i,t,\mathrm{mix}} - \eta\,\mathbf{g}_{i,t},
	\label{eq:update}
\end{equation}
where $\eta>0$ is a constant stepsize.

We  summarize the above operations in Algorithm~\ref{alg:etg-local}.

\begin{algorithm}[t]
	\caption{Event-Triggered Gossip: Local Procedure at Node $i$}
	\label{alg:etg-local}
	\DontPrintSemicolon
	\KwIn{Initial model $\mathbf{x}_{i,0}=\mathbf{x}_{0}, \mathbf{x}_{i,0}\in\!\mathbb{R}^d$; constant stepsize $\eta>0$; neighbor set $\mathcal{N}(i)$; mixing weights $\mathbf{W}$; thresholds $\{\tau_t\}_{t=0}^{T-1}$; rounds $T$.}
	\KwData{Snapshot $\hat{\mathbf{x}}_{i,0} \leftarrow \mathbf{x}_{i,0}$; caches $\tilde{\mathbf{x}}_{j\rightarrow i,-1} \leftarrow \mathbf{x}_{j,0}$ for all $j\in\mathcal{N}(i)$; set $\tilde{\mathbf{x}}_{i\rightarrow i,-1}\leftarrow \mathbf{x}_{i,0}$.}
	
	\For{$t=0,1,\ldots,T-1$}{
		
		\textbf{Trigger test \& broadcast on current model.}\;
		\If{$\|\mathbf{x}_{i,t}-\hat{\mathbf{x}}_{i,t}\|\ge \tau_t$}{
			broadcast the vector $\mathbf{x}_{i,t}$ to all $j\in\mathcal{N}(i)$;\quad
			$\hat{\mathbf{x}}_{i,t} \leftarrow \mathbf{x}_{i,t}$\;
		}
		
		\textbf{Receive window and cache update.}\;
		\For{$j\in\mathcal{N}(i)$}{
			\eIf{a message from $j$ is received in round $t$}{
				$\tilde{\mathbf{x}}_{j\rightarrow i,t} \leftarrow$ received $\mathbf{x}_{j,t}$
			}{
				$\tilde{\mathbf{x}}_{j\rightarrow i,t} \leftarrow \tilde{\mathbf{x}}_{j\rightarrow i,t-1}$
			}
		}
		$\tilde{\mathbf{x}}_{i\rightarrow i}^{\,t} \leftarrow \mathbf{x}_{i,t}$
		
		\textbf{Local stochastic gradient on current model.}\;
		Sample $\xi_{i,t}\sim\mathcal{D}_i$ and compute $\mathbf{g}_{i,t} = g_i(\mathbf{x}_{i,t};\xi_{i,t})$\;
		
		\textbf{One-step gossip mixing with available models.}\;
		$\mathbf{x}_{i,t,\mathrm{mix}} \leftarrow \displaystyle\sum_{j=1}^{n} \mathbf{W}_{ji}\,\tilde{\mathbf{x}}_{j\rightarrow i,t}$
		
		\textbf{Local update.}\;
		$\mathbf{x}_{i,t+1} \leftarrow \mathbf{x}_{i,t,\mathrm{mix}} - \eta\,\mathbf{g}_{i,t}$
	}
	\KwOut{Final average model  $\bar{\mathbf{x}}_T$.}
\end{algorithm}

\subsection{Obsolescence Error and One-Round Update}
In a global point of view, we  introduce the matrix form of the system update. 
At round $t$, we stack all local models column-wise as
\begin{equation}
	\mathbf{X}_{t} \;\triangleq\; 
	\big[\;\mathbf{x}_{1,t}\ \mathbf{x}_{2,t}\ \cdots\ \mathbf{x}_{n,t}\;\big]
	\;\in\;\mathbb{R}^{d\times n},
\end{equation}
where the $i$-th column is the model of node $i$ at round $t$. 
Similarly, we collect the stochastic gradients as
\begin{equation}
	\mathbf{G}_{t} \;\triangleq\;
	\big[\;\mathbf{g}_{1,t}\ \mathbf{g}_{2,t}\ \cdots\ \mathbf{g}_{n,t}\;\big]
	\;\in\;\mathbb{R}^{d\times n}.
\end{equation}

When mixing on $\mathbf{X}_{t}$, the cached model $\tilde{\mathbf{x}}_{j\rightarrow i,t}$ at node $i$ may differ from the true current model $\mathbf{x}_{j,t}$ due to event-triggered communication. 
We define the per-edge \emph{obsolescence error} as
\begin{equation}
	\mathbf{v}_{j\rightarrow i,t}\;\triangleq\;
	\tilde{\mathbf{x}}_{j\rightarrow i,t}-\mathbf{x}_{j,t},
	\qquad
	\mathbf{v}_{i\rightarrow i,t}=\mathbf{0}.
	\label{eq:edge-error}
\end{equation}
If node $j$ triggers communication at round $t$, then $\mathbf{v}_{j\rightarrow i,t}=\mathbf{0}, \forall i\in \mathcal{N}(j)\cup\{i\} $; otherwise,
\begin{equation}
	\|\mathbf{v}_{j\rightarrow i,t}\|\;\le\;\tau_{t},
	\qquad \forall\, i,j,
	\label{eq:edge-error-bound}
\end{equation}
since the cache stores a past broadcast of $j$ and because of the triggering rule.
Aggregating the obsolescence error in \eqref{eq:edge-error} across neighbors, the total error injected at node $i$ is
\begin{equation}
\mathbf{v}_{i,t}\;\triangleq\;\sum_{j=1}^{n}\mathbf{v}_{j\rightarrow i,t}\,\mathbf{W}_{ji}.\label{v_it}
\end{equation}
We collect $\mathbf{v}_{i,t}$ into an error matrix,
\begin{equation}
	\mathbf{V}_{t}\;\triangleq\;
	\big[\;\mathbf{v}_{1,t}\ \mathbf{v}_{2,t}\ \cdots\ \mathbf{v}_{n,t}\;\big]
	\;\in\;\mathbb{R}^{d\times n}.
\end{equation}
Now, we  obtain the matrix form of system update in \eqref{eq:update} as
\begin{equation}
	\mathbf{X}_{t+1}
	= \mathbf{X}_{t}\mathbf{W}
	- \eta\,\mathbf{G}_{t}
	+ \mathbf{V}_{t}.
	\label{eq:matrix-update}
\end{equation}
Equation~\eqref{eq:matrix-update} compactly summarizes the overall dynamics: 
the first term, $\mathbf{X}_{t}\mathbf{W}$, corresponds to gossip mixing; 
the second term, $-\eta\mathbf{G}_{t}$, is the local stochastic gradient update;
and the last term, $\mathbf{V}_{t}$, captures the perturbation induced by event-triggered communication\footnote{ When $\mathbf{V}_{t}=\mathbf{0}$ (i.e., all nodes communicate at every round),
\eqref{eq:matrix-update} reduces to the standard full-communication decentralized SGD update.}.

\section{Convergence Analysis}
In this section, we analyze the convergence  of the proposed event-triggered gossip algorithm. We first state the assumptions and characterize the dynamics of the network-wide average iterate. We then establish a unified ergodic convergence bound under general triggering thresholds and discuss its implications under different thresholding strategies.
\subsection{Assumptions for Convergence Analysis}

\begin{assumption}[Mixing Weights]\label{ass:mixing}
	The mixing matrix, $\mathbf{W}\in\mathbb{R}^{n\times n}$, is nonnegative, 
	{doubly stochastic and symmetric}, i.e.,
	\[
	\mathbf{W}\mathbf{1}=\mathbf{1}, \qquad \mathbf{1}^\top \mathbf{W}=\mathbf{1}^\top, \qquad \mathbf{W}=\mathbf{W}^\top.
	\]
	Let $\mathbf{J}\triangleq \frac{1}{n}\mathbf{1}\mathbf{1}^\top$. The spectral contraction factor is
	\[
	\delta \;\triangleq\; \|\mathbf{W}-\mathbf{J}\|_{2}^2 \;<\; 1.
	\]
\end{assumption}

\begin{assumption}[Smoothness]\label{ass:smooth}
	Each local objective function $f_i:\mathbb{R}^d \rightarrow \mathbb{R}$ is differentiable, 
	and its gradient $\nabla f_i(\cdot)$ is Lipschitz continuous with parameter $L>0$. 
	That is, for all $\mathbf{x},\mathbf{y}\in\mathbb{R}^d$ and for all $i\in[n]$,
	\[
	\|\nabla f_i(\mathbf{x})-\nabla f_i(\mathbf{y})\| \;\le\; L\,\|\mathbf{x}-\mathbf{y}\|.
	\]
\end{assumption}

\begin{assumption}[Bounded Variance]\label{ass:variance}
	For each node $i\in[n]$, the stochastic gradient 
	$\mathbf{g}_{i,t}=g_i(\mathbf{x}_{i,t};\xi_{i,t})$ 
	is an unbiased estimator of $\nabla f_i(\mathbf{x}_{i,t})$, 
	and its variance is uniformly bounded. Moreover, the divergence of local gradients from the global gradient is also bounded. Formally, for all $\mathbf{x}\in\mathbb{R}^d$,
	\[
	\mathbb{E}\big\|\mathbf{g}_{i,t}-\nabla f_i(\mathbf{x}_{i,t})\big\|^2 \;\le\; \alpha^2,
	\qquad \forall\, i\in[n],
	\]
	\[
	\mathbb{E}\big\|\nabla f_i(\mathbf{x})-\nabla f(\mathbf{x})\big\|^2 \;\le\; \beta^2,
	\qquad \forall\, i\in[n].
	\]
	The first expectation is taken with respect to the data sampling, 
	$\xi_{i,t}\sim\mathcal{D}_i$. The second expectation  is taken over a uniformly random choice  of node index $i \sim \mathcal{U}([n])$, which corresponds to measuring the average gradient 
discrepancy across nodes.
\end{assumption}

The above three assumptions are common in the literature on decentralized stochastic optimization, e.g., \cite{koloskova2019decentralized,wang2018cooperative,li2021decentralized,11077738,Zhou2025PreconditionedADMM}. 
In Assumption~\ref{ass:mixing},  symmetricity and double stochasticity  ensure that local updates are convex combinations of neighbor models and that the network-wide average is preserved. 
The spectral contraction factor $\delta=\|\mathbf{W}-\mathbf{J}\|_2<1$ further guarantees that the consensus error decays geometrically, 
which is equivalent to requiring that all eigenvalues of $\mathbf{W}$  have magnitude strictly less than $1$, except the largest one that is $1$. 
Assumption~\ref{ass:smooth} imposes $L$-smoothness on the local objectives, meaning that their gradients are Lipschitz continuous. 
This condition controls how fast the gradients can change with respect to the model parameters and is a standard requirement for analyzing the stability of gradient-based methods. 
Assumption~\ref{ass:variance} bounds both the stochastic gradient variance and the divergence across devices. 
$\alpha^2$ quantifies the upper limit of the noise introduced by local data sampling;  $\beta^2$ controls the level of heterogeneity among devices. 
Together, these bounds ensure that the stochasticity and non-independent and identically distributed (non-i.i.d.) effects remain well behaved and facilitate analysis.

\subsection{Global Average Dynamics}
Define the column-wise averages
\begin{equation}
	\bar{\mathbf{x}}_{t}\;\triangleq\; \mathbf{X}_{t}\frac{\mathbf{1}}{n}\in\mathbb{R}^{d},
	\qquad
	\bar{\mathbf{g}}_{t}\;\triangleq\; \mathbf{G}_{t}\frac{\mathbf{1}}{n}\in\mathbb{R}^{d},
\end{equation}
and the average perturbation
\begin{equation}
	\bar{\mathbf{e}}_{t}\;\triangleq\; \mathbf{V}_{t}\frac{\mathbf{1}}{n}\in\mathbb{R}^{d}.
\end{equation}
Multiplying the matrix update \eqref{eq:matrix-update}  by $\frac{\mathbf{1}}{n}$ and using the column–stochasticity $\,\mathbf{W}\mathbf{1}=\mathbf{1}\,$ in Assumption \ref{ass:mixing}, we obtain
\begin{equation}
		\bar{\mathbf{x}}_{\,t+1}
		= \bar{\mathbf{x}}_{\,t}\;-\;\eta\,\bar{\mathbf{g}}_{t}\;+\;\bar{\mathbf{e}}_{t}.
	\label{eq:avg-evolution}
\end{equation}
Hence, gossip mixing (by mixing matrix $\mathbf{W}$) cancels out in the average (it preserves $\bar{\mathbf{x}}_{t}$), and \emph{event-triggering} enters the average dynamics only through the additive perturbation $\bar{\mathbf{e}}_{t}$.

Recall $\mathbf{v}_{i,t}=\sum_{j=1}^{n}\mathbf{W}_{ji}\,\mathbf{v}_{j\rightarrow i,t}$ in \eqref{v_it}  and $\|\mathbf{v}_{j\rightarrow i,t}\|\le \tau_t$ in \eqref{eq:edge-error-bound}. Then, we have
\begin{align}
	\|\bar{\mathbf{e}}_{t}\|
	&= \Big\|\frac{1}{n}\sum_{i=1}^{n}\mathbf{v}_{i,t}\Big\|
	\;\le\; \frac{1}{n}\sum_{i=1}^{n}\|\mathbf{v}_{i,t}\| \notag\\
	&\le\; \frac{1}{n}\sum_{i=1}^{n}\sum_{j=1}^{n}\mathbf{W}_{ji}\,\|\mathbf{v}_{j\rightarrow i,t}\|
	\;\le\; \tau_t.
	\label{eq:avg-error-bound}
\end{align}
Consequently, the average iterate \eqref{eq:avg-evolution} behaves like a centralized SGD step with an \emph{additive perturbation} with magnitude  controlled by the trigger threshold $\tau_t$ (e.g., if $\tau_t\equiv 0$, then $\bar{\mathbf{e}}_t=\mathbf{0}$).

\subsection{Ergodic Convergence Result}
With the above assumptions and analyses, we now establish the convergence guarantee for the proposed event-triggered gossip algorithm in the following theorem with the proof provided in Appendix~\ref{appb}.

\begin{theorem}\label{thm1}
	Under Assumptions~1--3 and any stepsize $\eta>0$ such that
	\[
	\Gamma \triangleq 1-\frac{27n\eta^2L^2}{(1-\sqrt{\delta})^2}>0,\quad
	\Delta \;\triangleq\; \frac{1}{2}-\frac{27n\eta^2L^2}{(1-\sqrt{\delta})^2\Gamma} \;>\; 0,
	\]
	the iterates of Algorithm~1 satisfy the ergodic convergence bound
	\begin{align}
		&\frac{1}{T}\sum_{t=0}^{T-1}\mathbb{E}\big\|\nabla f(\bar{\mathbf{x}}_{t})\big\|^{2}\nonumber\\
		&\le \frac{f(\bar{\mathbf{x}}^{0})-f^{\ast}}{\eta\Delta\,T}
		+ \frac{\eta L^{2}}{\Gamma\Delta}\!\left(\frac{3n^{2}\alpha^{2}}{1-\delta}
		+ \frac{27\eta^{2}n\beta^{2}}{(1-\sqrt{\delta})^{2}}\right)
		+ \frac{\eta\alpha^{2}}{n\Delta} \nonumber\\
		&
		+ \frac{9\eta L^{2}}{(1-\sqrt{\delta})^{2}\Gamma\Delta}\cdot
		\frac{1}{T}\sum_{t=0}^{T-1} n\tau_{t}^{2}
		+ \frac{1}{\eta\Delta}\cdot \frac{1}{T}\sum_{t=0}^{T-1}\tau_{t}^{2}.
		\label{eq:ergodic-bound}
	\end{align}
\end{theorem}

Theorem~\ref{thm1} provides a unified ergodic convergence guarantee for the proposed event-triggered gossip algorithm under any feasible thresholds $\{\tau_t\}$; see \eqref{eq:ergodic-bound}. 
The bound makes explicit how the convergence  depends jointly on the stepsize $\eta$, the spectral contraction factor $(1-\sqrt{\delta})$, the stochastic gradient variance $\alpha^2$, the data heterogeneity $\beta^2$, and the triggering sequence $\{\tau_t\}$ through the stability constants $\Gamma$ and $\Delta$. 
In particular, the event-triggered mechanism injects perturbations into the average-iterate dynamics of $\bar{\mathbf{x}}_t$ via the terms controlled by $\{\tau_t\}$, as shown in \eqref{eq:avg-error-bound}. A smaller $\tau_t$ yields more frequent communication and tighter tracking of the centralized trajectory, whereas a larger $\tau_t$ reduces communication at the expense of a non-vanishing steady-state bias in \eqref{eq:ergodic-bound}.

\subsection{Triggering Threshold Discussion}
To further interpret \eqref{eq:ergodic-bound} and clarify the accuracy–communication trade-off, we specialize Theorem~\ref{thm1} to three representative regimes of the threshold sequence $\{\tau_t\}$. 
 \textbf{Case~A} corresponds to \emph{full communication} ($\tau_t \equiv 0$), which recovers the standard nonconvex rate; 
\textbf{Case~B} studies a \emph{constant threshold} ($\tau_t \equiv \tau>0$), which leads to a bounded steady-state bias determined by $\tau$; 
\textbf{Case~C} considers \emph{decaying thresholds} $\{\tau_t\}$, where the triggering condition becomes gradually looser over time.

\textbf{Case A (full communication, $\tau_t \equiv 0$).}
When every round involves communication ($\tau_t \equiv 0$), the ergodic bound in \eqref{eq:ergodic-bound} reduces to three dominant terms:
\[
\frac{1}{T}\sum_{t=0}^{T-1}\mathbb{E}\|\nabla f(\bar{\mathbf{x}}_t)\|^2
\;\le\;
\frac{K_1}{\eta T} \;+\; K_2\eta \;+\; K_3\eta^{3},
\]
where the constants are explicitly given by
\begin{align}
	K_1 &= \frac{f(\bar{\mathbf{x}}_0)-f^{*}}{\Delta}; \nonumber\\[3pt]
	K_2 &= \frac{L^2}{\Gamma\Delta}\!\left(\frac{3n^{2}\alpha^{2}}{1-\delta}\right) 
	+ \frac{\alpha^{2}}{n\Delta}; \nonumber\\[3pt]
	K_3 &= \frac{27L^{2}n\beta^{2}}{(1-\sqrt{\delta})^{2}\Gamma\Delta}. \nonumber
\end{align}

To ensure the validity of this bound, the stepsize, $\eta$, must satisfy three conditions simultaneously. 
The first condition is that the smoothness argument in the descent step requires $\eta \le 1/L$. 
The second condition is that the stability conditions, $\Gamma>0$ and $\Delta>0$, need to be guaranteed by the sufficient inequality $\tfrac{27n\eta^{2}L^{2}}{(1-\sqrt{\delta})^{2}} < \tfrac{1}{3}$, which yields
\[
\eta < \frac{1-\sqrt{\delta}}{9\sqrt{n}\,L}.
\]
The third condition is that, since $\eta\le 1/L$ and $\eta<\frac{1-\sqrt{\delta}}{9\sqrt{n}\,L}$,  $\eta\leq\eta_{\max}:=\min\{\tfrac{1}{L},\,\tfrac{1-\sqrt{\delta}}{9\sqrt{n}\,L}\}$. Using $\eta\le\eta_{\max}$, we squeeze the cubic term by $K_3\eta^3 \le K_3\eta_{\max}^2\,\eta$, so the right-hand side (RHS) is upper-bounded by $\frac{K_1}{\eta T}+\widetilde K_2\,\eta$ with $\widetilde K_2:=K_2+K_3\eta_{\max}^2$. Minimizing this upper bound gives $\eta^{*}=\sqrt{K_1/(\widetilde K_2 T)}=\Theta(T^{-1/2})$. Combining these conditions, a convenient stepsize choice is
\[
\eta = \Theta\Big(\min\Big\{\tfrac{1}{L},\ \tfrac{1-\sqrt{\delta}}{\sqrt{n}\,L},\ \tfrac{1}{\sqrt{T}}\Big\}\Big).
\]
With this setting, the algorithm achieves the standard convergence rate, as given by
\[
\frac{1}{T}\sum_{t=0}^{T-1}\mathbb{E}\|\nabla f(\bar{\mathbf{x}}_t)\|^2
= \mathcal{O}\!\left(\tfrac{1}{\sqrt{T}}\right).
\]
To this end, with full communication, the proposed algorithm achieves the same $\mathcal{O}(1/\sqrt{T})$ rate as centralized SGD \cite{ghadimi2013stochastic,bottou2018optimization,johnson2013accelerating} in the nonconvex setting, while the hidden constants in the convergence bound are affected by  the stochastic gradient variance ($\alpha>0$),  data heterogeneity ($\beta>0$), and  network spectral gap $(1-\sqrt{\delta})$.

\textbf{Case B (constant threshold, $\tau_t \equiv \tau>0$).}
Plugging $\tau_t\equiv\tau$ into \eqref{eq:ergodic-bound} yields, for any feasible stepsize
$\eta\in(0,\eta_{\max}]$ with $\eta_{\max}:=\min\{\tfrac{1}{L},\,\tfrac{1-\sqrt{\delta}}{9\sqrt{n}\,L}\}$,
\[
\frac{1}{T}\sum_{t=0}^{T-1}\mathbb{E}\|\nabla f(\bar{\mathbf{x}}_t)\|^2
\;\le\;
{\frac{K_1}{\eta T}}
\;+\;
{K_2\,\eta}
\;+\;
{\frac{K_3}{\eta}}
\;+\;
{K_4\,\eta^{3}},
\]
where
\begin{align}
	K_1 &= \frac{f(\bar{\mathbf{x}}_0)-f^{*}}{\Delta}; \nonumber\\[3pt]
	K_2 &= \frac{\alpha^{2}}{n\Delta}
	+ \frac{L^{2}}{\Gamma\Delta}\!\left(\frac{3n^{2}\alpha^{2}}{1-\delta}\right)
	+ \frac{9L^{2}n\tau^{2}}{(1-\sqrt{\delta})^{2}\Gamma\Delta}; \nonumber\\[3pt]
	K_3 &= \frac{\tau^{2}}{\Delta}; \qquad
	K_4 = \frac{27L^{2}n\beta^{2}}{(1-\sqrt{\delta})^{2}\Gamma\Delta}. \nonumber
\end{align}
Since $\eta\le\eta_{\max}$, we have
$K_4\,\eta^{3}\le K_4\,\eta_{\max}^{2}\,\eta$. Let $\widetilde K_2:=K_2+K_4\,\eta_{\max}^{2}$. Then,
\begin{align}
    \frac{1}{T}\sum_{t=0}^{T-1}\mathbb{E}\|\nabla f(\bar{\mathbf{x}}_t)\|^2
\;\le\; \frac{K_1}{\eta T}\;+\;\widetilde K_2\,\eta\;+\;\frac{K_3}{\eta}.\label{K3_n}
\end{align}
Minimizing the RHS of \eqref{K3_n} over $(0,\eta_{\max}]$ gives the implementable choice:
\[
{\eta}\;=\;\min\!\Big\{\sqrt{\frac{K_1/T+K_3}{\widetilde K_2}},\ \eta_{\max}\Big\},
\]
and, therefore,
\begin{align}
	\frac{1}{T}\sum_{t=0}^{T-1}\mathbb{E}\|\nabla f(\bar{\mathbf{x}}_t)\|^2
	&\;\le\; 2\sqrt{\widetilde K_2\!\left(\tfrac{K_1}{T}+K_3\right)} \nonumber\\[3pt]
	&\;=\; \mathcal{O}\!\Big(\tfrac{1}{\sqrt{T}}\Big)
	\;+\; 2\sqrt{\widetilde K_2\,K_3}. \nonumber
\end{align}
Hence, with a constant (non-decaying) threshold, the bound consists of a transient
$\mathcal{O}(T^{-1/2})$ term plus a non-vanishing steady-state bias $2\sqrt{\widetilde K_2\,K_3}$; the
ergodic gradient norm does not converge to $0$ as $T\to\infty$. The constants depend on the stochastic
gradient variance ($\alpha$),  data heterogeneity ($\beta$), and  spectral gap $(1-\sqrt{\delta})$
through $\Gamma$ and $\Delta$.

\textbf{Case C (decaying threshold, $\{\tau_t\}$).}
Let $\overline{\tau^2}:=\tfrac{1}{T}\sum_{t=0}^{T-1}\tau_t^2$. Plugging this into \eqref{eq:ergodic-bound} gives: 
\[
\frac{1}{T}\sum_{t=0}^{T-1}\mathbb{E}\|\nabla f(\bar{\mathbf{x}}_t)\|^2
\;\le\;
{\frac{K_1}{\eta T}}
\;+\;
{\big(K_2 + K_4\,\overline{\tau^2}\big)\,\eta}
\;+\;
{\frac{b_T}{\eta}}
\;+\;
{K_3\,\eta^{3}},
\]
for any feasible stepsize $\eta\in(0,\eta_{\max}]$ with $\eta_{\max}:=\min\{\tfrac{1}{L},\,\tfrac{1-\sqrt{\delta}}{9\sqrt{n}\,L}\}$, where
\begin{align}
	K_1 &= \frac{f(\bar{\mathbf{x}}^0)-f^{*}}{\Delta}; \nonumber\\
	K_2 &= \frac{\alpha^{2}}{n\Delta}
	+\frac{3n^{2}L^{2}\alpha^{2}}{(1-\delta)\Gamma\Delta}; \nonumber\\
	K_3 &= \frac{27L^{2}n\beta^{2}}
	{(1-\sqrt{\delta})^{2}\Gamma\Delta}; \nonumber\\
	K_4 &= \frac{9L^{2}n}
	{(1-\sqrt{\delta})^{2}\Gamma\Delta}; \qquad
	b_T = \frac{\overline{\tau^2}}{\Delta}. \nonumber
\end{align}
Since $\eta\le\eta_{\max}$, we have $K_3\eta^{3}\le K_3\eta_{\max}^{2}\,\eta$. Define 
\[
\widetilde a_T\;:=\;K_2\;+\;K_3\eta_{\max}^{2}\;+\;K_4\,\overline{\tau^2}.
\]
Then,
\begin{align}
    \frac{1}{T}\sum_{t=0}^{T-1}\mathbb{E}\|\nabla f(\bar{\mathbf{x}}_t)\|^2
\;\le\; \frac{K_1}{\eta T}\;+\;\widetilde a_T\,\eta\;+\;\frac{b_T}{\eta}.\label{bt_n}
\end{align}
Minimizing the RHS of \eqref{bt_n} over $(0,\eta_{\max}]$ yields the implementable choice:
\[
{\eta}\;=\;\min\!\Big\{\sqrt{\frac{K_1/T+b_T}{\widetilde a_T}},\ \eta_{\max}\Big\},
\]
and, therefore, the upper bound is given as
\begin{align}
\frac{1}{T}\sum_{t=0}^{T-1}\mathbb{E}\|\nabla f(\bar{\mathbf{x}}_t)\|^2
\;\le\; 2\sqrt{\widetilde a_T\!\left(\frac{K_1}{T}+b_T\right)}.\label{22}
\end{align}

There are two scenarios that need to be considered. Scenario (i) If $\tau_t=\tau_0/\sqrt{t+1}$, then $\overline{\tau^2}=\Theta\!\big(\tfrac{\log T}{T}\big)$, so
$b_T=\Theta\!\big(\tfrac{\log T}{\Delta\,T}\big)$ and $\widetilde a_T=(K_2+K_3\eta_{\max}^{2})+\Theta\!\big(\tfrac{\log T}{T}\big)$.
Defining $K_2' := K_2 + K_3\eta_{\max}^{2}$, we can write 
$\widetilde a_T = K_2' + \Theta\!\big(\tfrac{\log T}{T}\big)$. 
Note that $\tfrac{K_1}{T}=o(\tfrac{\log T}{T})$ and hence $\tfrac{K_1}{T}+b_T=\Theta(\tfrac{\log T}{T})$, while $\widetilde a_T=\Theta(1)$.
The RHS of \eqref{22} scales as 
$2\sqrt{\Theta(1)\cdot\Theta(\tfrac{\log T}{T})}
=\mathcal{O}\!\big(\sqrt{\tfrac{\log T}{T}}\big)$,
which yields
\[
\frac{1}{T}\sum_{t=0}^{T-1}\mathbb{E}\|\nabla f(\bar{\mathbf{x}}_t)\|^2
=\mathcal{O}\!\Big(\sqrt{\tfrac{\log T}{T}}\Big).
\]
Ignoring the slowly varying logarithmic factor, the convergence rate can be described in the soft-O form, as $\tilde{\mathcal{O}}(T^{-1/2})$.

Scenario (ii) If $\tau_t=\tau_0/(t+1)$, then $\overline{\tau^2}=\Theta\!\big(\tfrac{1}{T}\big)$, so
$b_T=\Theta\!\big(\tfrac{1}{\Delta\,T}\big)$ and $\widetilde a_T=(K_2+K_3\eta_{\max}^{2})+\Theta\!\big(\tfrac{1}{T}\big)$.
By defining $K_2' := K_2 + K_3\eta_{\max}^{2}$, we can write 
$\widetilde a_T = K_2' + \Theta\!\big(\tfrac{1}{T}\big)$. 
Note that $\tfrac{K_1}{T}=\Theta(\tfrac{1}{T})$ and $b_T=\Theta(\tfrac{1}{T})$, and hence $\tfrac{K_1}{T}+b_T=\Theta(\tfrac{1}{T})$, while $\widetilde a_T=\Theta(1)$.
The RHS of~\eqref{22} scales as 
$2\sqrt{\Theta(1)\cdot\Theta(\tfrac{1}{T})}
=\mathcal{O}\!\big(T^{-1/2}\big)$,
which yields
\[
\frac{1}{T}\sum_{t=0}^{T-1}\mathbb{E}\|\nabla f(\bar{\mathbf{x}}_t)\|^2
=\mathcal{O}\!\big(T^{-1/2}\big).
\]
This corresponds to the standard stochastic nonconvex convergence rate $\mathcal{O}(T^{-1/2})$ of centralized SGD.
In both scenarios, the hidden constants depend on the stochastic gradient variance ($\alpha$), the data heterogeneity ($\beta$),
and the spectral gap $(1-\sqrt{\delta})$ through $\Gamma$ and $\Delta$.

    The above discussion reveals that the triggering threshold $\{\tau_t\}$ plays a
decisive role in shaping both the communication cost and the achievable
accuracy. When $\tau_t \equiv 0$, every round involves communication and the
algorithm reduces to standard distributed learning, recovering the
$O(T^{-1/2})$ convergence rate of centralized SGD. When $\tau_t$ is a fixed
positive constant, communication is substantially reduced, but the ergodic
gradient norm converges only up to a non-vanishing steady-state bias
proportional to $\tau$, reflecting persistent model discrepancies. In contrast,
when $\tau_t$ decays over time (e.g., $\tau_t = \Theta(t^{-1/2})$ or 
$\tau_t = \Theta(t^{-1})$), the injected perturbation gradually diminishes and the
full $\mathcal{O}(T^{-1/2})$ rate is recovered. This indicates that properly decaying
thresholds enable the proposed event-triggered gossip method to match the
theoretical performance of full-communication decentralized SGD while
significantly reducing inter-node transmissions.

\section{Simulation Results}
In this section, we evaluate the performance of the proposed event-triggered gossip algorithm through extensive simulations. We first describe the experimental settings and datasets, and then examine the impact of the triggering threshold and network sparsity on learning accuracy and communication cost. Finally, we compare the proposed method with  state-of-the-art distributed learning schemes.
\subsection{Simulation Settings}

We consider a distributed learning system consisting of $M=20$ edge devices connected through an undirected communication graph.  
The learning task is image classification on the MNIST and Fashion-MNIST datasets.  
MNIST contains grayscale handwritten digit images (0–9) and is widely regarded as a relatively simple benchmark,
whereas Fashion-MNIST consists of grayscale clothing images from 10 apparel categories and is generally more challenging to learn due to higher visual complexity.

Among all samples, $10{,}000$ are used for training and $10{,}000$ for testing.  
Since the MNIST and  Fashion-MNIST  datasets contain 10 distinct classes, we partition the devices into 10 groups of equal size, where each group is assigned a dataset corresponding to one specific class. This setup ensures that the network experiences data heterogeneity, as described in~\cite{mcmahan2017communication}. The network sparsity is defined as the proportion of zero elements in the mixing matrix, specifically, the number of zero elements in $\mathbf{W}$ divided by \( n^2 \). Unless otherwise specified, we conduct simulation under network sparsity of  $0.3$.

Each local model is implemented as a  convolutional neural network (CNN) with the following structure:
two convolutional layers with $5\times5$ kernels and channel sizes $(1,10)$ and $(10,20)$,  
followed by a batch-normalization (or dropout) layer and two fully connected layers with dimensions $320\!-\!50\!-\!10$.  
The ReLU activation function is applied after each layer. The network outputs a 10-dimensional vector representing class logits.
All models are initialized identically and trained using SGD with a constant learning rate $\eta=0.02$.  
In the proposed event-triggered gossip algorithm, each device $i$ decides whether to communicate at round $t$ according to the triggered threshold
\begin{align}
	\tau_{t} = \epsilon \|\mathbf{x}_{0}\|,
\end{align}
where $\epsilon>0$ is a small relative threshold coefficient and $\mathbf{x}_{0}$ is the initial model for all devices.  The results presented  are based on the average of 30 Monte Carlo experiments.
\begin{figure}[t]
	\centering
	\includegraphics[width=0.9\linewidth]{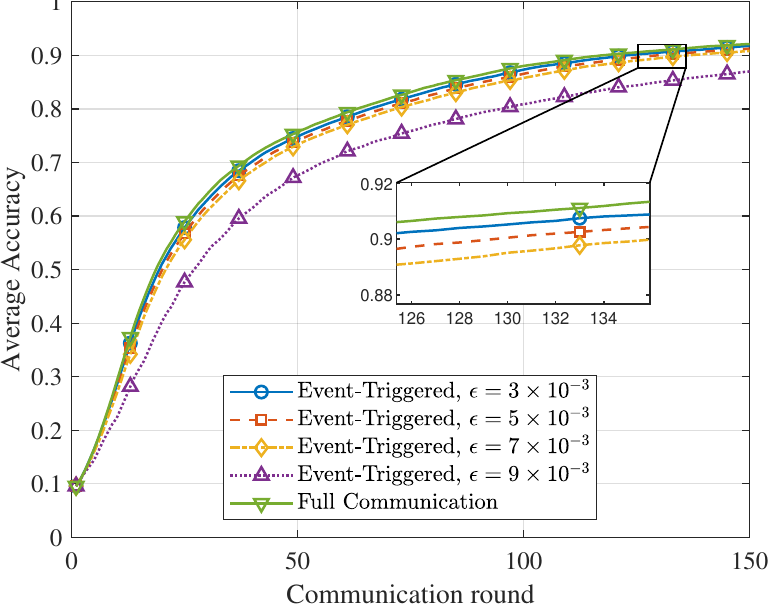}
	\caption{Test accuracy of the averaged model under different triggering thresholds $\epsilon$ (MNIST).}
	\label{fig:mnist_diff_epsilon_acc}
\end{figure}
\begin{figure}[t]
	\centering
	\includegraphics[width=0.9\linewidth]{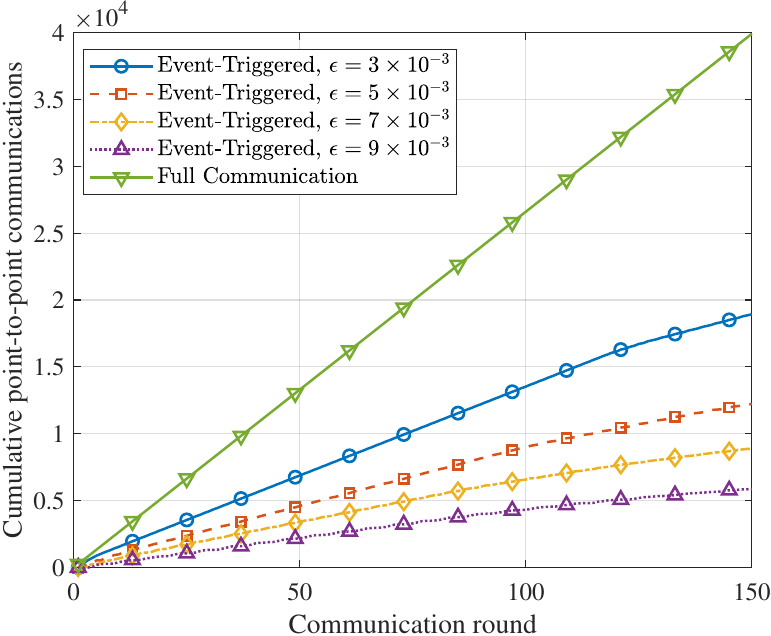}
	\caption{Cumulative point-to-point communication volume versus $\epsilon$ (MNIST).}
	\label{fig:mnist_diff_epsilon_comm}
\end{figure}
\begin{figure}[t]
	\centering
	\includegraphics[width=0.9\linewidth]{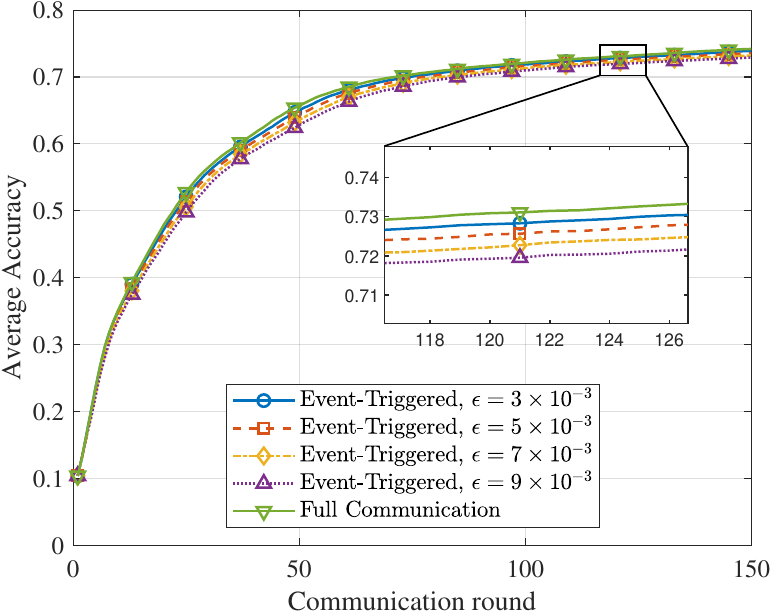}
	\caption{Test accuracy of the averaged model under different triggering thresholds $\epsilon$ (Fashion-MNIST).}
	\label{fig:fmnist_diff_epsilon_acc}
\end{figure}
\begin{figure}[t]
	\centering
	\includegraphics[width=0.9\linewidth]{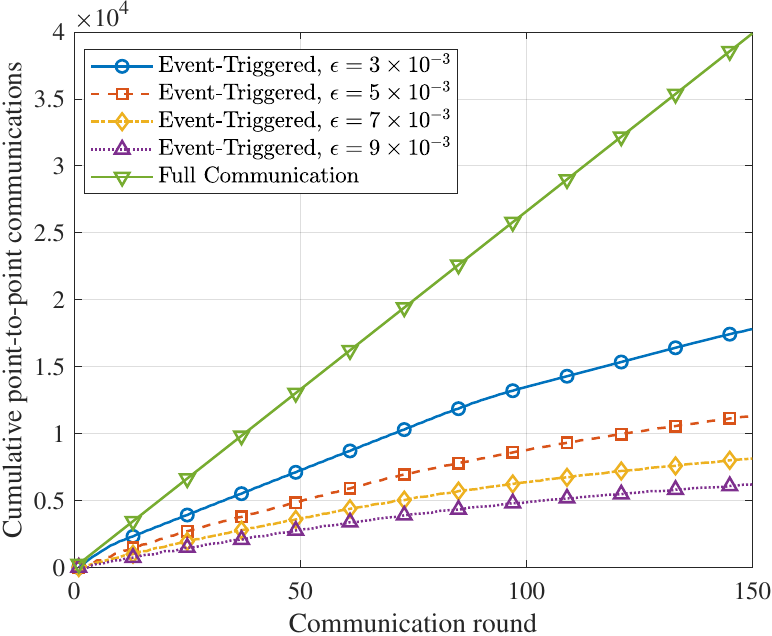}
	\caption{Cumulative point-to-point communication volume versus $\epsilon$ (Fashion-MNIST).}
	\label{fig:fmnist_diff_epsilon_comm}
\end{figure}

\subsection{Impact of Trigger Threshold on Accuracy and Communication}

Figs.~\ref{fig:mnist_diff_epsilon_acc}--\ref{fig:mnist_diff_epsilon_comm} 
and Figs.~\ref{fig:fmnist_diff_epsilon_acc}--\ref{fig:fmnist_diff_epsilon_comm}
show the impact of the triggering threshold $\epsilon$ on the test performance 
and communication cost for the MNIST and Fashion-MNIST datasets, respectively.  
The test accuracy is evaluated on the {averaged model}, i.e.,
\[
\bar{\mathbf{x}}_{t} \triangleq \frac{1}{M}\sum_{i=1}^{M}\mathbf{x}_{i,t}.
\]
The communication cost is measured by the total number of point-to-point \emph{one-way} transmissions among all devices; that is, each model message sent from a node to one of its neighbors is counted as one transmission. 
The {full communication} case corresponds to $\tau_t \equiv 0$, 
meaning that all devices exchange local models with their neighbors at every round.

As shown in Fig.~\ref{fig:mnist_diff_epsilon_acc}, 
the accuracy of the averaged model $\bar{\mathbf{x}}_{t}$ on MNIST 
gradually approaches that of the full communication case as $\epsilon$ decreases.  
A smaller $\epsilon$ makes the triggering condition, i.e.,
\[
\|\mathbf{x}_{i,t}-\hat{\mathbf{x}}_{i,t}\|\ge\tau_{t}=\epsilon\|\mathbf{x}_{0}\|,
\]
easier to satisfy, which results in more frequent inter-device communications, 
faster consensus, and higher accuracy.  
Correspondingly, Fig.~\ref{fig:mnist_diff_epsilon_comm} shows that 
the cumulative communication volume increases monotonically as $\epsilon$ decreases.
When $\epsilon$ is large, triggers are rare, leading to minimal communication but slower convergence;
as $\epsilon\!\to\!0$, the behavior converges to the full communication case, 
achieving the highest accuracy at the cost of the largest overhead.

The experiments on Fashion-MNIST, reported in 
Figs.~\ref{fig:fmnist_diff_epsilon_acc} and \ref{fig:fmnist_diff_epsilon_comm}, 
exhibit a similar trend.  
As $\epsilon$ decreases, the accuracy consistently improves while the total communication volume increases.  
However, the performance gap between different thresholds is smaller on Fashion-MNIST than on MNIST, 
indicating that the proposed event-triggered mechanism maintains stable accuracy even under sparser communication on more complex datasets.

\begin{table}[t]
	\centering
	\caption{Average Accuracy and Cumulative Point-to-Point Communication Volume After 150 Rounds}
	\label{tab:epsilon_comparison}
	\renewcommand{\arraystretch}{1.2}
	\setlength{\tabcolsep}{4pt}
	\begin{tabular}{lcccc}
		\toprule
		\multirow{2}{*}{Scheme} 
		& \multicolumn{2}{c}{MNIST} 
		& \multicolumn{2}{c}{Fashion-MNIST} \\ 
		\cmidrule(lr){2-3} \cmidrule(lr){4-5}
		& Accuracy & Comm. & Accuracy & Comm.  \\ 
		\midrule
		Full  Comm.      & 0.9219 & 39900 & 0.7422& 39900 \\
	$\epsilon=3\times10^{-3}$ & 0.9179& 18934 & 0.7393 & 17811 \\
		 $\epsilon=5\times10^{-3}$ & 0.9135 &  12230& 0.7360 & 11328 \\
		$\epsilon=7\times10^{-3}$ & 0.9087 & 8893 & 0.7325 & 8134 \\
		 $\epsilon=9\times10^{-3}$ & 0.8711 & 5853 & 0.7292 & 6203 \\
		\bottomrule
	\end{tabular}
	\footnotesize{Notes: “Comm.” is the cumulative number of point-to-point transmissions. 
	Results are averaged over trials, unless otherwise specified.}
\end{table}

Let $\mathrm{CR}\!=\!1-\tfrac{\mathrm{Comm}(\epsilon)}{\mathrm{Comm}(\text{Full})}$ denote the communication reduction ratio, 
and  $\mathrm{AD}
= \mathrm{Acc}(\mathrm{Full}) - \mathrm{Acc}(\epsilon)$ denote the accuracy drop relative to full communication (in percentage points).
In Table~\ref{tab:epsilon_comparison}, our event-triggered scheme achieves a marked reduction in point-to-point transmissions with only a marginal loss in accuracy:

\begin{itemize}
	\item \textbf{Key operating point ($\epsilon=5\times10^{-3}$).} 
	On {MNIST}, the proposed method reaches an average accuracy of $0.9135$ ($\mathrm{AD}=0.84$\, percentage points (pp), i.e., $0.91\%$ relative) 
	while using \textbf{$\mathrm{CR}=69.35\%$} fewer transmissions (\(12230/39900\)).  
	On {Fashion-MNIST}, it attains $0.7360$ (drop $0.62$\,pp, $0.84\%$ relative) with \textbf{$\mathrm{CR}=71.61\%$} fewer transmissions (\(11328/39900\)).  
	These results verify that our event-triggered gossip can preserve the performance of full communication with only about ${1}/{3}$ of the communication cost.
	
	\item \textbf{More aggressive saving ($\epsilon=7\times10^{-3}$).} 
	MNIST reduces transmissions by \textbf{$77.71\%$} (8893 vs. 39900) with an accuracy drop of only $1.32$\,pp ($1.43\%$).  
	Fashion-MNIST reduces by \textbf{$79.61\%$} (8134 vs. 39900) with a $0.97$\,pp ($1.31\%$) drop.
	
	\item \textbf{Conservative setting ($\epsilon=3\times10^{-3}$).} 
	Even with a smaller threshold, MNIST still saves \textbf{$52.55\%$} communication with merely $0.40$\,pp ($0.43\%$) loss;  
	Fashion-MNIST saves \textbf{$55.36\%$} with a $0.29$\,pp ($0.39\%$) loss.
\end{itemize}

Overall, the event-triggered mechanism provides a tunable accuracy--communication trade-off: 
As $\epsilon$ decreases, accuracy approaches that of full communication while communication cost smoothly increases.  
At $\epsilon=5\times10^{-3}$, the method achieves \emph{near-identical} performance to full communication yet consumes only \textbf{30.65\%} (MNIST) and \textbf{28.39\%} (Fashion-MNIST) of its point-to-point transmissions, i.e., communication savings of \textbf{69.35\%} and \textbf{71.61\%}, respectively.

\subsection{Performance Analysis under Different Network Sparsity}

To investigate the performance of the proposed scheme under varying network configurations, we conducted experiments with different levels of network sparsity. 

In Figs.~\ref{fig:mnist_diff_sparsity_acc} and~\ref{fig:mnist_diff_sparsity_comm}, we observe that: as the sparsity increases, both the point-to-point communication volume and the accuracy of our algorithm are affected. 
In Fig.~\ref{fig:mnist_diff_sparsity_comm}, we can see a noticeable decrease in the cumulative communication volume as the sparsity increases. This is because  each node has fewer neighbors to communicate with in sparser topologies, leading to fewer point-to-point transmissions when the communication is triggered. Consequently, as sparsity increases, the number of active communication links decreases, resulting in a reduction in overall communication volume.

On the other hand, Fig.~\ref{fig:mnist_diff_sparsity_acc} demonstrates that the accuracy of the algorithm decreases as the network sparsity increases. This performance degradation can be attributed to the non-i.i.d. data partitioning deployed in the experiments. As the sparsity grows, the communication between nodes becomes more sporadic, meaning that nodes are less able to collaborate effectively and synchronize their model parameters. This leads to slower consensus, as fewer nodes communicate in each round, and the model update process becomes less efficient. The lack of communication and collaboration among nodes hampers the model's ability to converge effectively.

In general, while higher sparsity reduces the communication overhead, it  diminishes the accuracy of the algorithm due to the decreased collaboration between nodes. 

\begin{figure}[t]
	\centering
	\includegraphics[width=0.9\linewidth]{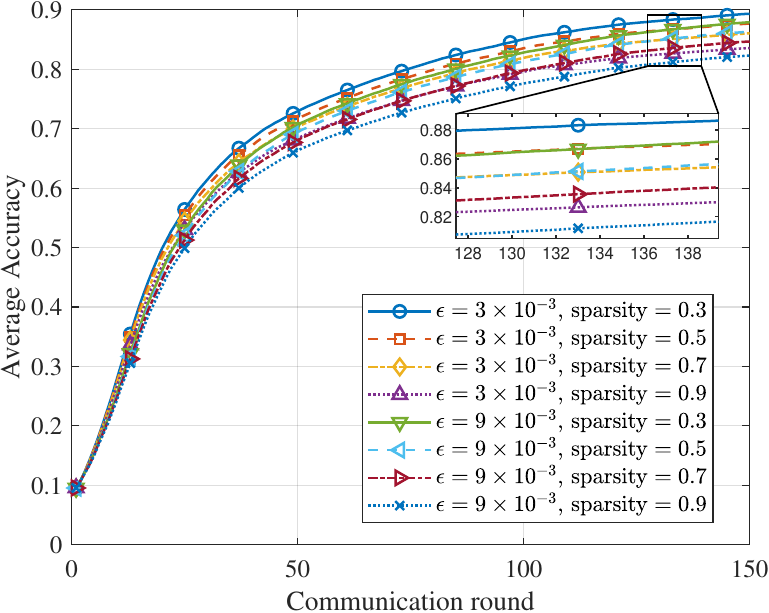}
	\caption{Average accuracy under different levels of network sparsity.}
	\label{fig:mnist_diff_sparsity_acc}
\end{figure}

\begin{figure}[t]
	\centering
	\includegraphics[width=0.9\linewidth]{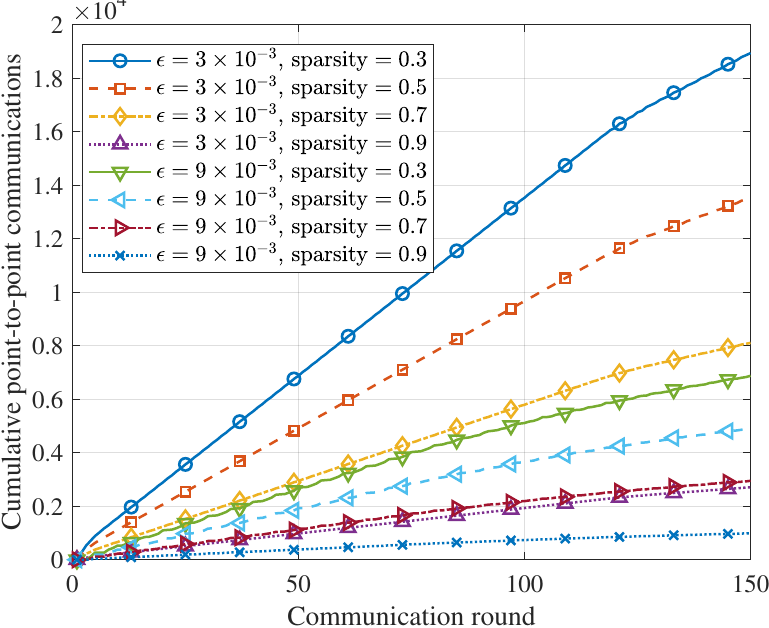}
	\caption{Cumulative point-to-point communication volume under different levels of network sparsity.}
	\label{fig:mnist_diff_sparsity_comm}
\end{figure}

\subsection{Comparison with State-of-the-art Schemes}

To comprehensively evaluate the effectiveness of the proposed event-triggered gossip algorithm, we compare it with  representative distributed learning schemes that represent different communication mechanisms. These schemes include both conventional and state-of-the-art strategies widely adopted in distributed learning.

\textbf{1) Full Communication \cite{koloskova2020unified}:}  
This scheme performs standard gossip averaging in every iteration \cite{koloskova2020unified}, where all devices exchange information with all of their neighbors. It serves as an upper-bound baseline, representing the best achievable performance with maximum communication overhead, corresponding to the case $\tau_t \equiv 0$.

\textbf{2) Event-Triggered Gossip (Proposed):}  
In the proposed method, each device triggers communication only when the deviation between its current model $\mathbf{x}_{i,t}$ and its latest broadcast snapshot $\hat{\mathbf{x}}_{i,t}$ exceeds a relative threshold, i.e., 
\[
\|\mathbf{x}_{i,t}-\hat{\mathbf{x}}_{i,t}\|\ge \tau_{t}=\epsilon\|\mathbf{x}_{0}\|.
\]
Only the triggered devices update their snapshots and communicate with neighbors. 

\textbf{3) Periodic Gossip \cite{haddadpour2019local}:}  
In this time-triggered scheme, devices perform global gossip communication once every $K_p$ rounds and conduct only local updates in-between \cite{haddadpour2019local}. The periodic exchange effectively reduces communication frequency but may slow down model consensus, especially when $K_p$ is large.

\textbf{4) Probabilistic Gossip \cite{zhai2025decentralized}:}
In this scheme, each directed communication link $(i,j)$ is activated
independently according to a Bernoulli random variable. At round $t$, node $i$ transmits its model to neighbor $j$ with
probability $p_{ij}$, i.e., the transmission event is
$a_{ij,t} \sim \mathrm{Bernoulli}(p_{ij})$. If $a_{ij,t}=1$, the message
is sent; if $a_{ij,t}=0$, no transmission occurs and node $j$ reuses its
own model. See \cite{zhai2025decentralized} for details.

\textbf{5) Variable Working Nodes \cite{7462298}:}  
Inspired by the variable node activation strategy proposed in \cite{7462298}, this scheme allows each node to be activated with probability $p_k$ at each round. Only activated nodes perform both communication and gradient updates within the active subgraph, whereas inactive nodes remain unchanged. This joint control of communication and computation further reduces system overhead at the expense of slower convergence in highly sparse activation regimes.

These five schemes cover a broad range of distributed learning communication paradigms, including event-triggered, time-triggered, and probabilistic triggering mechanisms.

We first provide one representative parameter configuration for
each scheme to give a concrete comparison. In this setting, we use
$\epsilon=0.005$ for event-triggered gossip, $K_p=5$ for periodic
gossip, and $p_{{ij}}=0.5, \forall i,j$ and $p_k=0.3$ for the
probabilistic and variable working-node schemes. The corresponding
performance after 150 training rounds is summarized in
Table~\ref{tab:scheme_mnist_fmnist}.
Here, the proposed
event-triggered gossip achieves comparable accuracy  to full
communication while reducing the total point-to-point transmissions
by more than a factor of three on both datasets. Periodic gossip
and variable working-node schemes consume fewer communications but
suffer noticeable accuracy degradation, whereas probabilistic gossip
retains reasonable accuracy but requires significantly more
communications than the proposed method.

To  visualize how accuracy varies with the
communication budget, we  sweep the key
communication-related parameters of each scheme and depict their
accuracy--communication trade-offs. For each scheme, we sweep its primary
communication-control hyperparameters over a range of values
(e.g., the relative triggering threshold $\epsilon$ for
event-triggered gossip, the communication period
$K_{\mathrm{p}}$ for periodic gossip, and the activation
probabilities $p_{{ij}}, \forall i,j$ and $p_k$ for the probabilistic
and variable-working schemes). For each configuration, we run
training for $T=150$ rounds and record the pair consisting of
the final test accuracy of the averaged model and the
cumulative number of point-to-point transmissions among all
devices. The resulting accuracy--communication pairs are
plotted in Fig.~\ref{fig:pareto-mnist} for MNIST and
Fig.~\ref{fig:pareto-fmnist} for Fashion-MNIST. In these plots,
desirable operating points lie towards the upper-left corner,
corresponding to high accuracy with low communication volume.

On the MNIST dataset
(Fig.~\ref{fig:pareto-mnist}), the full-communication scheme
naturally achieves the highest accuracy but also incurs the
largest communication volume, placing its operating points near
the upper-right corner. In contrast, the proposed
event-triggered gossip algorithm generates a suite of points
that cluster close to the upper-left region, forming an
approximate Pareto frontier. For a wide range of target
accuracies (e.g., above $0.90$), event-triggered gossip attains
comparable performance to full communication while requiring
significantly fewer point-to-point transmissions. In particular,
for a given communication budget, the proposed method achieves
higher accuracy than the periodic gossip and variable-working
schemes, whose points lie further down and/or to the right,
indicating either slower convergence or pronounced accuracy
loss under aggressive communication reduction. Probabilistic
gossip yields a set of intermediate trade-offs, but its points
are  dominated by those of the proposed method in
the accuracy--communication plane.

A similar pattern is observed on the Fashion-MNIST dataset
(Fig.~\ref{fig:pareto-fmnist}). The event-triggered gossip
scheme again occupies the upper-left region of the plot and
remains close to the full-communication accuracy while
substantially reducing communication. Periodic gossip and
variable-working schemes exhibit lower accuracy for comparable
communication volumes, reflecting the adverse impact of
infrequent synchronization and partial node activation on
consensus quality. Probabilistic gossip provides moderate
savings but still requires noticeably more transmissions than
the proposed method to reach the same accuracy level.

These the Pareto-style results  demonstrate that
the proposed event-triggered gossip algorithm offers the most
favorable accuracy--communication trade-off among all
considered schemes. Across both MNIST and Fashion-MNIST, its
operating points consistently lie on or near the empirical
Pareto frontier, enabling substantial reductions in
communication volume while preserving almost the same accuracy
as the full-communication baseline.

\begin{table}[t]
	\centering
	\caption{Average Accuracy and Cumulative Point-to-Point Communication Volume After 150 Rounds}
	\label{tab:scheme_mnist_fmnist}
	\renewcommand{\arraystretch}{1.2}
	\setlength{\tabcolsep}{5pt}
	\begin{tabular}{lcccc}
		\toprule
		\multirow{2}{*}{Scheme} 
		& \multicolumn{2}{c}{MNIST} 
		& \multicolumn{2}{c}{Fashion-MNIST} \\
		\cmidrule(lr){2-3}\cmidrule(lr){4-5}
		& Accuracy & Comm. & Accuracy & Comm. \\
		\midrule
		Full communication         & 0.9219 & 39900 & 0.7422 & 39900 \\
		Event-triggered    & 0.9135 & 12230 & 0.7360 & 11328 \\
		Periodic gossip         & 0.8454 &  7980 & 0.6884 &  7980 \\
		Probabilistic gossip      & 0.9134 & 19904 & 0.7357 & 19904 \\
		Variable working     & 0.7406 & 11939 & 0.6430 & 11939 \\
		\bottomrule
	\end{tabular}
	
	\vspace{2pt}

\end{table}

\begin{figure}[t]
  \centering
  \includegraphics[width=0.9\linewidth]{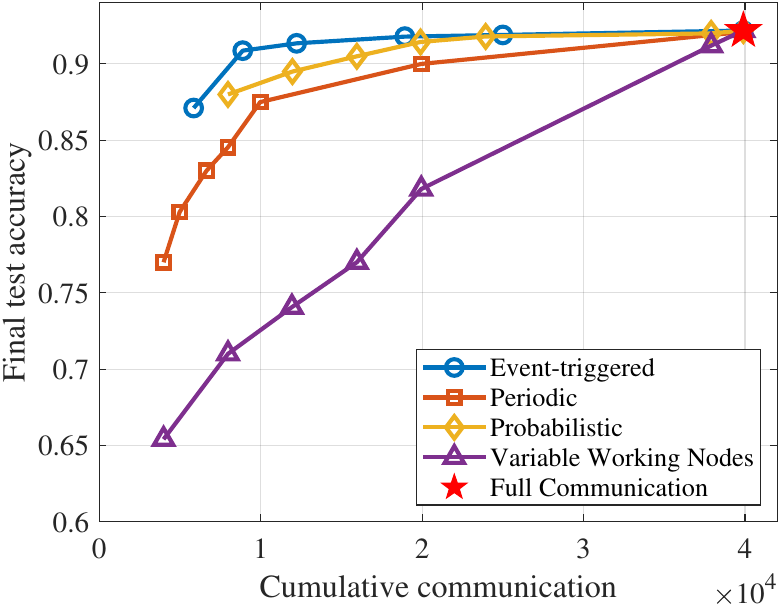}
  \caption{Test accuracy versus cumulative point-to-point
  communication volume for different distributed learning schemes
  on the MNIST dataset.}
  \label{fig:pareto-mnist}
\end{figure}

\begin{figure}[t]
  \centering
  \includegraphics[width=0.9\linewidth]{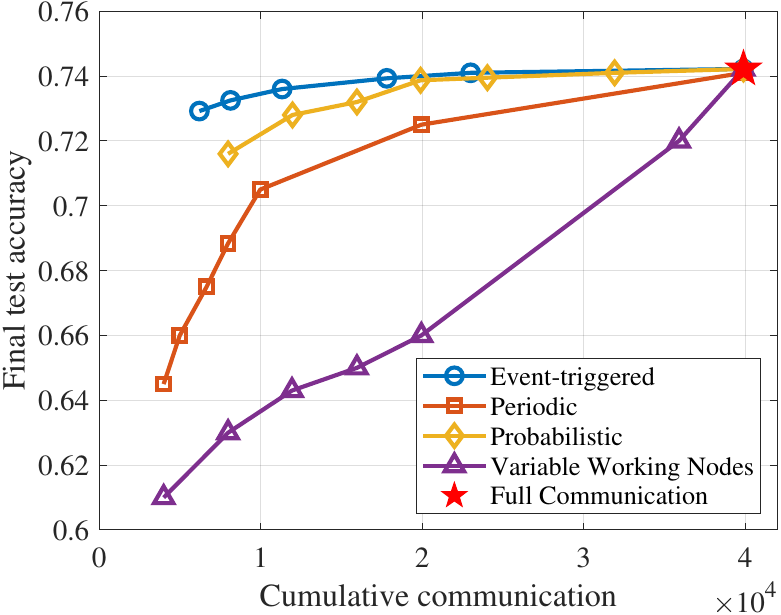}
  \caption{Test accuracy versus cumulative point-to-point
  communication volume for different distributed learning schemes
  on the Fashion-MNIST dataset. }
  \label{fig:pareto-fmnist}
\end{figure}

\section{Conclusions}
In this paper, we developed a distributed learning framework based on an event-triggered gossip mechanism, which allows each device to autonomously decide when to communicate with its neighbors according to local model deviations. 
We conducted a rigorous convergence analysis and discussed the convergence properties of the proposed scheme under different triggering thresholds, providing theoretical insights into the trade-off between its accuracy and communication efficiency. 
Extensive experiments on MNIST and Fashion-MNIST verified that the proposed framework can drastically reduce point-to-point communication while maintaining high learning accuracy.

	\bibliographystyle{IEEEtran}
	\bibliography{IEEEabrv,mybib}
    	\begin{IEEEbiography}[{\includegraphics[width=1in,height =1.25in,clip,keepaspectratio]{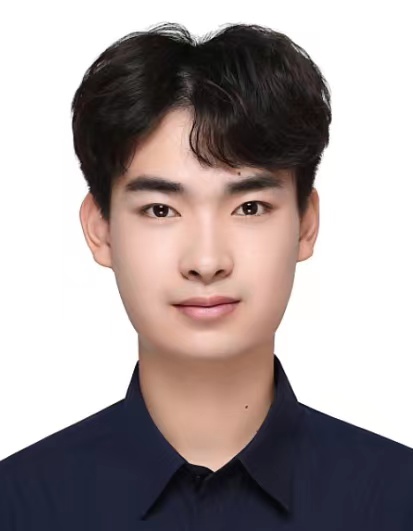}}]
		{Zhiyuan Zhai} received the B.S. degree in Communication Engineering from the School of Information and Communication Engineering, University of Electronic Science and Technology of China, in 2022. He is currently pursuing the Ph.D. degree  with the Department of Communication Science and Engineering, Fudan University, Shanghai, China.
		His research interests include  machine learning, signal processing and mobile edge computing.
	\end{IEEEbiography}

	\begin{IEEEbiography}[{\includegraphics[width=1in,height =1.25in,clip,keepaspectratio]{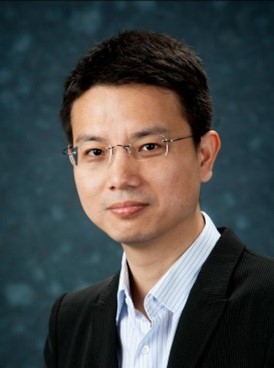}}]
		{Xiaojun Yuan} (Fellow, IEEE) received the PhD degree in electrical engineering from the City University
of Hong Kong, in 2009. From 2009 to 2011, he was
a research fellow with the Department of Electronic
Engineering, City University of Hong Kong. He was
also a visiting scholar with the Department of Electrical Engineering, University of Hawaii at Manoa, in
spring and summer 2009, as well as in the same period
of 2010. From 2011 to 2014, he was a research assistant professor with the Institute of Network Coding,
The Chinese University of Hong Kong. From 2014 to
2017, he was an assistant professor with the School of Information Science
and Technology, ShanghaiTech University. He is now a professor with the
National Key Laboratory of Wireless Communications, University of Electronic
Science and Technology of China. His research interests cover a broad range of
signal processing, machine learning, and wireless communications, including
but not limited to intelligent communications, structured signal reconstruction,
Bayesian approximate inference, distributed learning, etc. He has published
more than 320 peer-reviewed research papers in the leading international journals
and conferences in the related areas. He has served on many technical programs
for international conferences. He was an editor of IEEE leading journals, including IEEE Transactions on Wireless Communications and IEEE Transactions on
Communications. He was a co-recipient of IEEE Heinrich Hertz Award 2022,
and a co-recipient of IEEE Jack Neubauer Memorial Award 2025.
	\end{IEEEbiography}

    \begin{IEEEbiography}[{\includegraphics[width=1in,height=1.25in,clip,keepaspectratio]{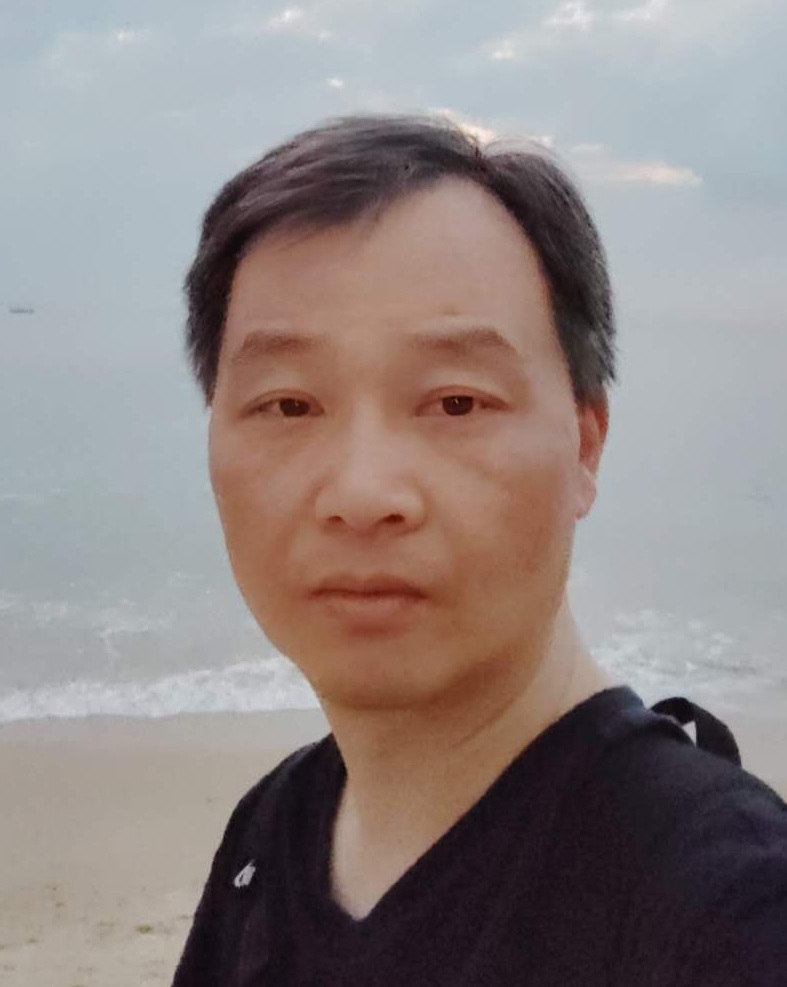}}]
{Wei Ni} (M’09-SM’15-F’24) received the B.E. and Ph.D. degrees in Electronic Engineering from Fudan University, Shanghai, China, in 2000 and 2005, respectively. He is the Associate Dean (Research) in the School of Engineering, Edith Cowan University, Perth, and a Conjoint Professor at the University of New South Wales, Sydney, Australia. He is also a Technical Expert at Standards Australia with a focus on the international standardization of Big Data and AI. He was a Deputy Project Manager at the Bell Labs, Alcatel/Alcatel-Lucent from 2005 to 2008; a Senior Research Engineer at Nokia from 2008 to 2009; and a Senior Principal Research Scientist and Group Leader at the Commonwealth Scientific and Industrial Research Organisation (CSIRO) from 2009 to 2025. His research interest lies in distributed and trusted learning with constrained resources, quantum Internet, and their applications to system efficiency, integrity, and resilience. He is a co-recipient of the ACM Conference on Computer and Communications Security (CCS) 2025 Distinguished Paper Award, and four Best Paper Awards. He has been an Editor for IEEE Transactions on Wireless Communications since 2018, IEEE Transactions on Vehicular Technology since 2022, IEEE Transactions on Information Forensics and Security and IEEE Communication Surveys and Tutorials since 2024, and IEEE Transactions on Network Science and Engineering and IEEE Transactions on Cloud Computing since 2025. He was Chair of the IEEE VTS NSW Chapter (2020 -- 2022), Track Chair for VTC-Spring 2017, Track Co-chair for IEEE VTC-Spring 2016, Publication Chair for BodyNet 2015, and Student Travel Grant Chair for WPMC 2014.
\end{IEEEbiography}
	
	\begin{IEEEbiography}[{\includegraphics[width=1in,height =1.25in,clip,keepaspectratio]{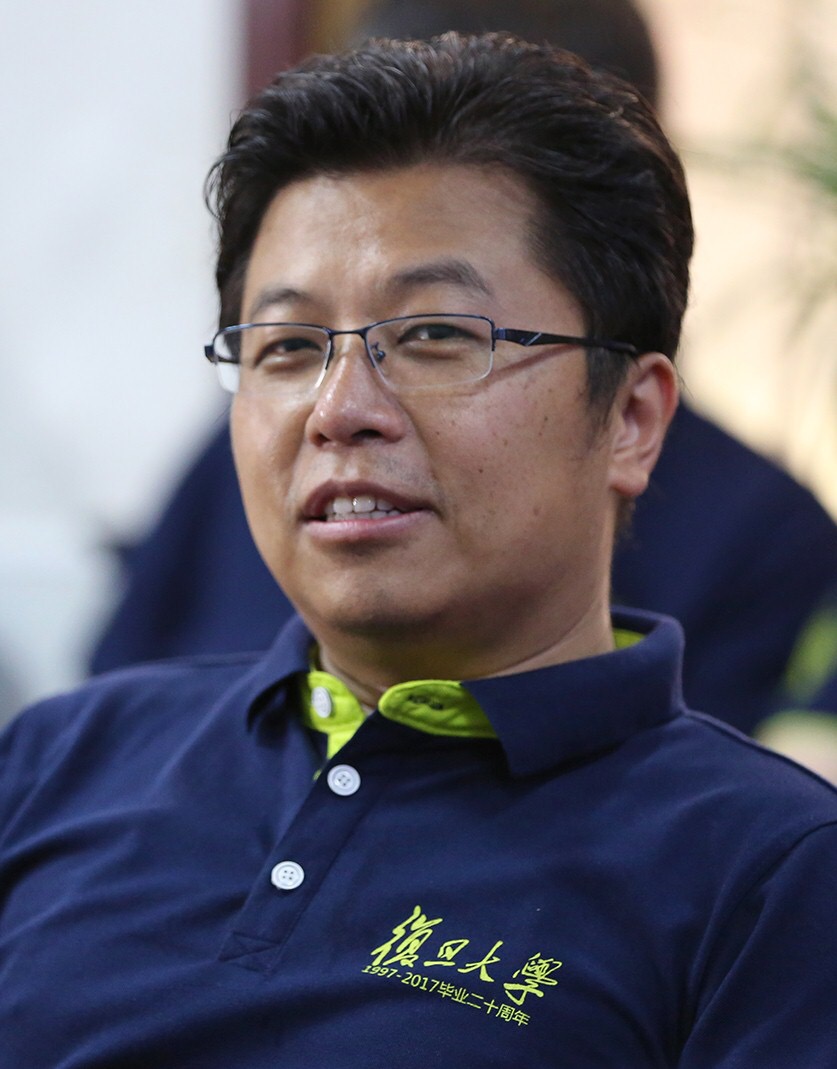}}]
		{Xin Wang} (Fellow, IEEE)  received the BSc and MSc
degrees from Fudan University, Shanghai, China, in
1997 and 2000, respectively, and the PhD degree from
Auburn University, Auburn, Alabama, in 2004, all
in electrical engineering. From September 2004 to
August 2006, he was a post-doctoral research associate with the Department of Electrical and Computer
Engineering, University of Minnesota, Minneapolis.
In August 2006, he joined the Department of Electrical Engineering, Florida Atlantic University, Boca
Raton, Florida, as an assistant professor, then was
promoted to a tenured associate professor, in 2010. He is currently a distinguished professor and the chair of the Department of Communication Science
and Engineering, Fudan University. His research interests include stochastic
network optimization, energyefficient communications, cross-layer design, and
signal processing for communications. He is a member of the Signal Processing
for Communications and Networking Technical Committee of IEEE Signal
Processing Society. He is a senior area editor of the IEEE Transactions on Signal
Processing and an editor of the IEEE Transactions on Wireless Communications.
In the past, he served as an associate editor for the IEEE Transactions on Signal
Processing, an editor for the IEEE Transactions on Vehicular Technology, and
an associate editor for the IEEE Signal Processing Letters. He is a distinguished
speaker of the IEEE Vehicular Technology Society.
	\end{IEEEbiography}

\begin{IEEEbiography}[{\includegraphics[width=1in,height =1.25in,clip,keepaspectratio]{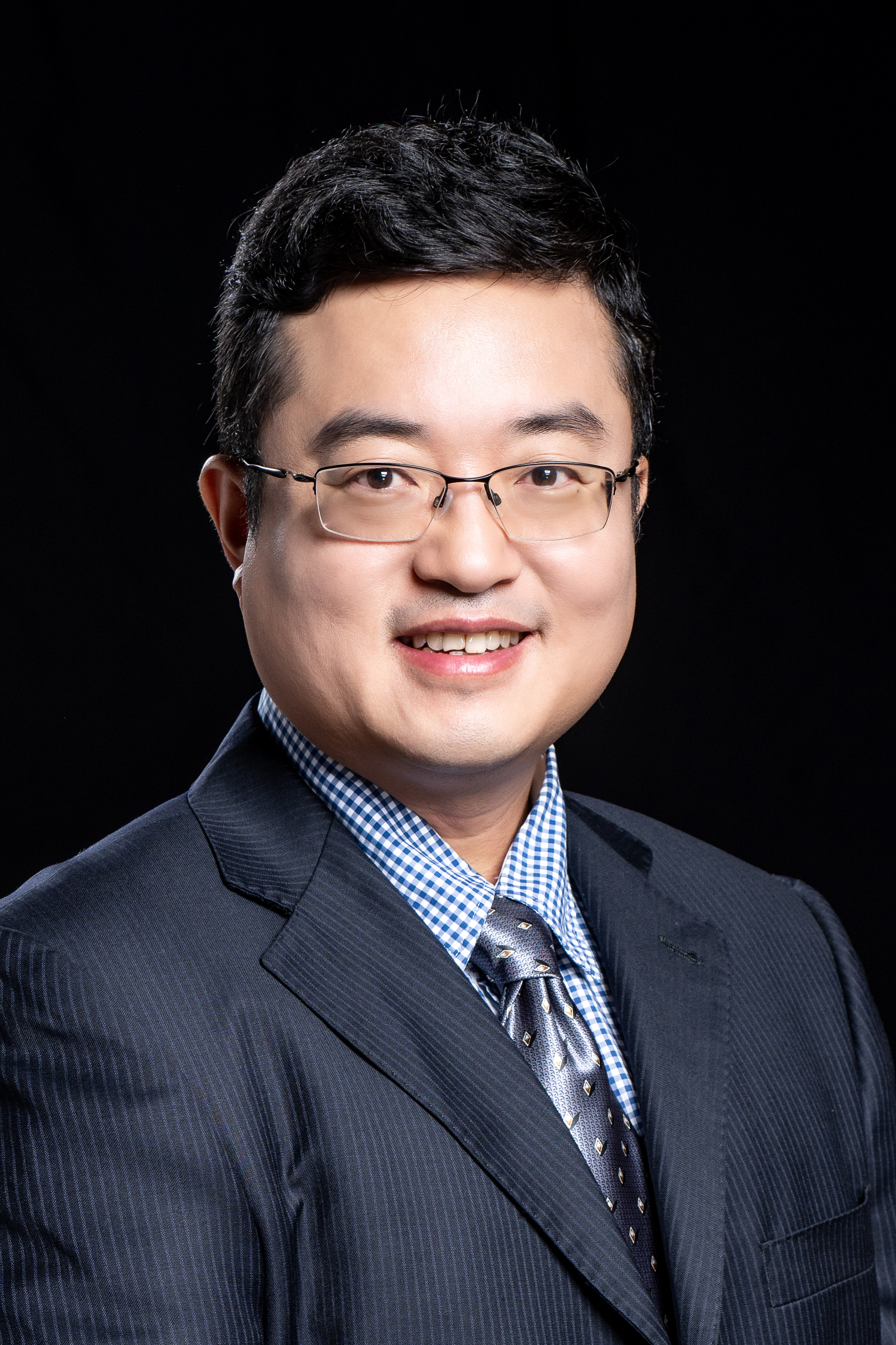}}] {Rui Zhang} (S'00-M'07-SM'15-F'17) received the B.Eng. (first-class Hons.) and M.Eng. degrees from the National University of Singapore, Singapore, and the Ph.D. degree from the Stanford University, Stanford, CA, USA, all in electrical engineering.
 
From 2007 to 2009, he worked as a research scientist at the Institute for Infocomm Research, ASTAR, Singapore. In 2010, he joined the Department of Electrical and Computer Engineering of National University of Singapore, where he is now a Provost’s Chair Professor. He is also an Adjunct Professor with the School of Science and Engineering, The Chinese University of Hong Kong, Shenzhen, China. He has published over 600 papers, all in the field of wireless communications and networks. He has been listed as a Highly Cited Researcher by Thomson Reuters/Clarivate Analytics since 2015. His current research interests include intelligent surfaces, reconfigurable antennas, radio mapping, non-terrestrial communications, wireless power transfer, AI and optimization methods.      

He was the recipient of the 6th IEEE Communications Society Asia-Pacific Region Best Young Researcher Award in 2011, the Young Researcher Award of National University of Singapore in 2015, the Wireless Communications Technical Committee Recognition Award in 2020, the IEEE Signal Processing and Computing for Communications (SPCC) Technical Recognition Award in 2021, and the IEEE Communications Society Technical Committee on Cognitive Networks (TCCN) Recognition Award in 2023. His works received 18 IEEE Best Journal Paper Awards, including the IEEE Marconi Prize Paper Award in Wireless Communications in 2015 and 2020, the IEEE Signal Processing Society Best Paper Award in 2016, the IEEE Communications Society Heinrich Hertz Prize Paper Award in 2017, 2020 and 2022, the IEEE Communications Society Stephen O. Rice Prize in 2021, etc. He served for over 30 international conferences as the TPC co-chair or an organizing committee member. He was an elected member of the IEEE Signal Processing Society SPCOM Technical Committee from 2012 to 2017 and SAM Technical Committee from 2013 to 2015. He served as the Vice Chair of the IEEE Communications Society Asia-Pacific Board Technical Affairs Committee from 2014 to 2015, a member of the Steering Committee of the IEEE Wireless Communications Letters from 2018 to 2021, a member of the IEEE Communications Society Wireless Communications Technical Committee (WTC) Award Committee from 2023 to 2025. He was a Distinguished Lecturer of IEEE Signal Processing Society and IEEE Communications Society from 2019 to 2020. He served as an Editor for several IEEE journals, including the IEEE TRANSACTIONS ON WIRELESS COMMUNICATIONS from 2012 to 2016, the IEEE JOURNAL ON SELECTED AREAS IN COMMUNICATIONS: Green Communications and Networking Series from 2015 to 2016, the IEEE TRANSACTIONS ON SIGNAL PROCESSING from 2013 to 2017, the IEEE TRANSACTIONS ON GREEN COMMUNICATIONS AND NETWORKING from 2016 to 2020, and the IEEE TRANSACTIONS ON COMMUNICATIONS from 2017 to 2022. He now serves as an Editorial Board Member of npj Wireless Technology, and the Chair of the IEEE Communications Society Wireless Communications Technical Committee (WTC) Award Committee. He is a Fellow of the Academy of Engineering Singapore. \end{IEEEbiography}

	\begin{IEEEbiography}[{\includegraphics[width=1in,height =1.25in,clip,keepaspectratio]{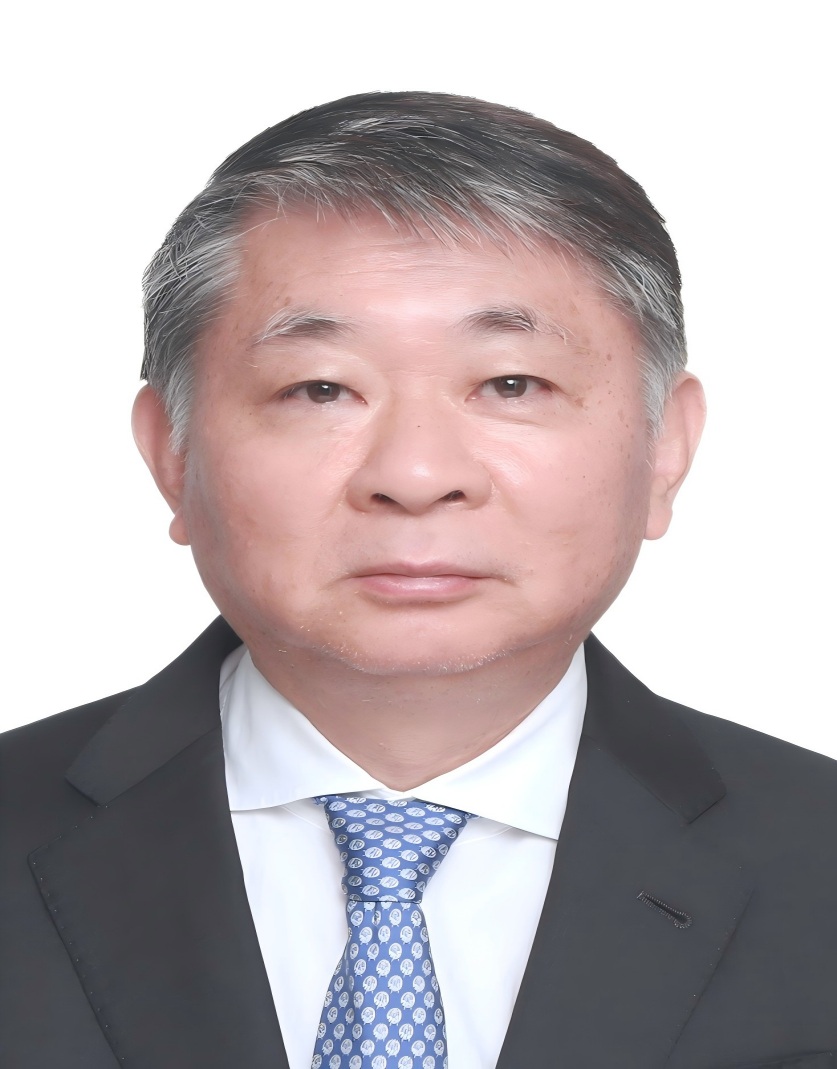}}]
		{Geoffrey Ye Li} (Fellow, IEEE)  is currently a Chair Professor at Imperial College London, UK.  Before joining Imperial in 2020, he was with Georgia Tech and AT\&T Labs – Research (previous Bell Labs) for 25 years in total. He made fundamental contributions to orthogonal frequency division multiplexing (OFDM) for wireless communications, established a framework on resource cooperation in wireless networks, and pioneered deep learning for communications. In these areas, he has published over 700 journal and conference papers in addition to over 40 granted patents. His publications have been cited over 85,000 times with an H-index over 130. He has been listed as a Highly Cited Researcher by Clarivate/Web of Science almost every year.
Dr. Geoffrey Ye Li was elected to Fellow of the Royal Academic of Engineering (FREng), IEEE Fellow, and IET Fellow for his contributions to signal processing for wireless communications. He received 2024 IEEE Eric E. Sumner Award, 2019 IEEE ComSoc Edwin Howard Armstrong Achievement Award, and several other awards from IEEE Signal Processing, Vehicular Technology, and Communications Societies.
\end{IEEEbiography}

\clearpage

\appendices

\section{Proof of Theorem 1}\label{appb}
{\textbf{Step 1}: Average iterate.}
Since $\mathbf{W}\mathbf{1}=\mathbf{1}$, multiplying the update 
$\mathbf{X}_{t+1}=\mathbf{X}_t\mathbf{W}-\eta\,\mathbf{G}_t+\mathbf{V}_t$
by $\mathbf{1}/n$ gives
\begin{equation}
	\bar{\mathbf{x}}_{t+1}
	= \bar{\mathbf{x}}_{t}
	- \eta\,\bar{\mathbf{g}}_{\,t}
	+ \bar{\mathbf{e}}_{\,t},
	\quad
	\bar{\mathbf{g}}_{\,t}= \mathbf{G}_t\frac{\mathbf{1}}{n},\ \
	\bar{\mathbf{e}}_{\,t}= \mathbf{V}_t\frac{\mathbf{1}}{n}.
	\label{eq:avg-dyn}
\end{equation}

{{\textbf{Step 2}: One-step descent on $f(\cdot)$.}
Assume $0<\eta\le 1/L$. Using $\bar{\mathbf{x}}_{t+1}=\bar{\mathbf{x}}_t-\eta(\bar{\mathbf{g}}_{\,t}-\bar{\mathbf{e}}_{\,t}/\eta)$,
\begin{align}
	&\quad\,\,\mathbb{E}f(\bar{\mathbf{x}}_{t+1})\nonumber\\
	&= \mathbb{E}f\big(\bar{\mathbf{x}}_t-\eta(\bar{\mathbf{g}}_{t}-\bar{\mathbf{e}}_{t}/\eta)\big) \nonumber\\
	&\overset{(a)}{\le} \mathbb{E}f(\bar{\mathbf{x}}_t)
	- \eta\,\mathbb{E}\!\left\langle \nabla f(\bar{\mathbf{x}}_t),\,\bar{\mathbf{g}}_{t}-\bar{\mathbf{e}}_{t}/\eta\right\rangle
+ \frac{L\eta^2}{2}\,\mathbb{E}\big\|\bar{\mathbf{g}}_{t}-\bar{\mathbf{e}}_{t}/\eta\big\|^2 \nonumber\\
	&\overset{(b)}{=} \mathbb{E}f(\bar{\mathbf{x}}_t)
	- \frac{\eta}{2}\,\mathbb{E}\|\nabla f(\bar{\mathbf{x}}_t)\|^2
	- \frac{\eta}{2}\,\mathbb{E}\big\|\bar{\mathbf{g}}_{t}-\bar{\mathbf{e}}_{t}/\eta\big\|^2 \nonumber\\
	& {}+ \frac{\eta}{2}\,\mathbb{E}\Big\|\nabla f(\bar{\mathbf{x}}_t)-\big(\bar{\mathbf{g}}_{t}-\bar{\mathbf{e}}_{t}/\eta\big)\Big\|^2
	+ \frac{L\eta^2}{2}\,\mathbb{E}\big\|\bar{\mathbf{g}}_{t}-\bar{\mathbf{e}}_{t}/\eta\big\|^2
	\nonumber\\
	&\overset{(c)}{\le} \mathbb{E}f(\bar{\mathbf{x}}_t)
	- \frac{\eta}{2}\,\mathbb{E}\|\nabla f(\bar{\mathbf{x}}_t)\|^2
	+ \frac{\eta}{2}\,\mathbb{E}\big\|\nabla f(\bar{\mathbf{x}}_t)-\bar{\mathbf{g}}_{t}+\bar{\mathbf{e}}_{t}/\eta\big\|^2
	\nonumber\\
	&\overset{(d)}{\le} \mathbb{E}f(\bar{\mathbf{x}}_t)
	- \frac{\eta}{2}\mathbb{E}\|\nabla f(\bar{\mathbf{x}}_t)\|^2
	\!+ \!\eta{\mathbb{E}\|\nabla f(\bar{\mathbf{x}}_t)-\bar{\mathbf{g}}_{t}\|^2}
	\!+ \frac{1}{\eta}\mathbb{E}\|\bar{\mathbf{e}}_{t}\|^2\!\!\notag
	\\&= \mathbb{E}f(\bar{\mathbf{x}}_t)
	- \frac{\eta}{2}\mathbb{E}\|\nabla f(\bar{\mathbf{x}}_t)\|^2
	\!+ \!\eta T_1
	\!+ \frac{1}{\eta}\mathbb{E}\|\bar{\mathbf{e}}_{t}\|^2\!\!
	\label{eq:desc-core}
\end{align}
where 
($a$) is due to the $L$-smoothness, i.e.,
$f(y)\le f(x)+\langle\nabla f(x),y-x\rangle+\frac{L}{2}\|y-x\|^2$,
($b$) is based on  $2\langle a,b\rangle=\|a\|^2+\|b\|^2-\|a-b\|^2$, 
($c$) is due to  $\eta\le 1/L$, and 
($d$) is based on $\|u+v\|^2\le 2\|u\|^2+2\|v\|^2$.

{{\textbf{Step 3}: Bounding $T_1$.}
Recall $T_1=\mathbb{E}\|\nabla f(\bar{\mathbf{x}}_t)-\bar{\mathbf{g}}_{\,t}\|^2$ with 
$\bar{\mathbf{x}}_t=\mathbf{X}_t\frac{\mathbf{1}}{n}$ and 
$\bar{\mathbf{g}}_{\,t}=\mathbf{G}_t\frac{\mathbf{1}}{n}
= \frac{1}{n}\sum_{i=1}^n \mathbf{g}_{i,t}$.
Add and subtract the average true gradient:
\begin{align}
	T_1
	&= \mathbb{E}\Big\|\!
	{\nabla f(\bar{\mathbf{x}}_t)\!-\!\frac{1}{n}\sum_{i=1}^n \nabla f_i(\mathbf{x}_{i,t})}
	\!+\! {\frac{1}{n}\sum_{i=1}^n \!\big(\nabla f_i(\mathbf{x}_{i,t})-\mathbf{g}_{i,t}\big)}\!
	\Big\|^2 \nonumber\\
	&= \mathbb{E}\|A_t\|^2 + \mathbb{E}\|B_t\|^2 + 2\,\mathbb{E}\langle A_t,B_t\rangle,
	\label{eq:T1-split}
\end{align}
where $\nabla f(\bar{\mathbf{x}}_t)-\frac{1}{n}\sum_{i=1}^n \nabla f_i(\mathbf{x}_{i,t})=A_t$, and $\frac{1}{n}\sum_{i=1}^n(\nabla f_i(\mathbf{x}_{i,t})-\mathbf{g}_{i,t})= B_t$.
By unbiasedness, $\mathbb{E}[\mathbf{g}_{i,t}]=\nabla f_i(\mathbf{x}_{i,t})$, hence
$\mathbb{E}[B_t]=\mathbf{0}$ and the cross term vanishes:
$\mathbb{E}\langle A_t,B_t\rangle=\mathbb{E}\langle A_t,\mathbb{E}[B_t]\rangle=0$.
Therefore,
\begin{equation}
	T_1 \;=\; \mathbb{E}\|A_t\|^2 + \mathbb{E}\|B_t\|^2.
	\label{eq:T1-two}
\end{equation}

For $\mathbb{E}\|B_t\|^2$, we have
\begin{IEEEeqnarray}{rCl}
	\mathbb{E}\|B_t\|^2
	&=& \mathbb{E}\Big\|\frac{1}{n}\sum_{i=1}^n\big(\nabla f_i(\mathbf{x}_{i,t})-\mathbf{g}_{i,t}\big)\Big\|^2 \nonumber\\
	&\le& \frac{1}{n^2}\sum_{i=1}^n \mathbb{E}\big\|\nabla f_i(\mathbf{x}_{i,t})-\mathbf{g}_{i,t}\big\|^2
	\;\overset{(a)}{\le}\; \frac{\alpha^2}{n},
	\label{eq:T1-var}
\end{IEEEeqnarray}
where  ($a$) is based on Assumption~\ref{ass:variance}.


{
	Combining \eqref{eq:desc-core} with the decomposition 
	$\mathbb{E}\|\nabla f(\bar{\mathbf{x}}_t)-\bar{\mathbf{g}}_{\,t}\|^2
	= T_2 + \frac{\alpha^2}{n}$, where
	\[
	T_2 \;\triangleq\; \mathbb{E}\Big\|\nabla f(\bar{\mathbf{x}}_t)
	- \frac{1}{n}\textstyle\sum_{i=1}^n \nabla f_i(\mathbf{x}_{i,t})\Big\|^2,
	\]
	we obtain
\begin{multline}
	\mathbb{E}f(\bar{\mathbf{x}}_{t+1})
	\le \mathbb{E}f(\bar{\mathbf{x}}_{t})
	-\frac{\eta}{2}\,\mathbb{E}\|\nabla f(\bar{\mathbf{x}}_{t})\|^{2}
	+ \eta\,T_{2} \\
	+ \frac{\eta\alpha^{2}}{n}
	+ \frac{1}{\eta}\,\mathbb{E}\|\bar{\mathbf{e}}_{t}\|^{2}.
	\label{eq:iterate-ineq}
\end{multline}
	The last term on the RHS of \eqref{eq:iterate-ineq} captures the average perturbation due to the event-triggered communication in round $t$.
	
{{\textbf{Step 4}: Bounding $T_2$ via consensus error.}
	By Jensen's inequality and the $L$-smoothness of $\{f_i\}$, it follows that
	\begin{align}
		T_2
		&= \mathbb{E}\Big\| \frac{1}{n}\textstyle\sum_{i=1}^n \big(\nabla f_i(\bar{\mathbf{x}}_t)-\nabla f_i(\mathbf{x}_{i,t})\big)\Big\|^2\nonumber\\
		&\le \frac{1}{n}\sum_{i=1}^n \mathbb{E}\big\|\nabla f_i(\bar{\mathbf{x}}_t)-\nabla f_i(\mathbf{x}_{i,t})\big\|^2 \nonumber\\
		&\overset{(a)}{\le} \frac{L^2}{n}\sum_{i=1}^n \mathbb{E}\big\|\bar{\mathbf{x}}_t-\mathbf{x}_{i,t}\big\|^2
		\overset{(a)}{=} \frac{L^2}{n}\,\mathbb{E}\big\|\mathbf{X}_t(\mathbf{I}-\mathbf{J})\big\|_F^2,
		\label{eq:T2-consensus}
	\end{align}
	where $\mathbf{J}=\frac{1}{n}\mathbf{1}\mathbf{1}^\top$, ($a$) is due to the $L$-smoothness assumption, and ($b$) is due to 
	$\sum_{i=1}^n\|\mathbf{x}_{i,t}-\bar{\mathbf{x}}_t\|^2=\|\mathbf{X}_t(\mathbf{I}-\mathbf{J})\|_F^2$.
	
    {{\textbf{Step 5}: Unrolling the consensus error $Q_{t,i}$.}
	Define the per-node disagreement energy
	\[
	Q_{t,i}\ \triangleq\ \mathbb{E}\big\|\bar{\mathbf{x}}_t-\mathbf{x}_{i,t}\big\|^2
	= \mathbb{E}\big\|\mathbf{X}_t\frac{\mathbf{1}}{n}-\mathbf{X}_t\mathbf{e}_i\big\|^2,
	\]
	with $\mathbf{e}_i$ the $i$-th canonical basis.
	Unrolling the update $\mathbf{X}_{s+1}=\mathbf{X}_s\mathbf{W}-\eta\,\mathbf{G}_s+\mathbf{V}_s$ yields
	\begin{equation}
		\mathbf{X}_t
		= \mathbf{X}_0\mathbf{W}^t
		- \eta\sum_{j=0}^{t-1}\mathbf{G}_j\,\mathbf{W}^{\,t-1-j}
		+ \sum_{j=0}^{t-1}\mathbf{V}_j\,\mathbf{W}^{\,t-1-j}.
		\label{eq:unroll}
	\end{equation}
	Hence,
	\begin{align}
		\bar{\mathbf{x}}_t-\mathbf{x}_{i,t}
		&= \mathbf{X}_0\mathbf{a}_{t,i}
		- \eta\sum_{j=0}^{t-1}\mathbf{G}_j\mathbf{a}_{t-1-j,i}
		+ \sum_{j=0}^{t-1}\mathbf{V}_j\mathbf{a}_{t-1-j,i}.
		\label{eq:barxi-diff}
	\end{align}
	If the initialization is in consensus (e.g., $\mathbf{X}_0(\mathbf{I}-\mathbf{J})=\mathbf{0}$, which includes $\mathbf{X}_0=\mathbf{0}$ as a special case), the first term on the RHS of \eqref{eq:barxi-diff} vanishes. By
	$\|a{+}b{+}c\|^2\le 3(\|a\|^2{+}\|b\|^2{+}\|c\|^2)$, we obtain
	\begin{align}
		&Q_{t,i}
		\le 3\eta^2\,{\mathbb{E}\Big\|\sum_{j=0}^{t-1}\big(\mathbf{G}_j-\partial f(\mathbf{X}_j)\big)\mathbf{a}_{t-1-j,i}\Big\|^2}
		\nonumber\\
		& + 3\eta^2\,{\mathbb{E}\Big\|\sum_{j=0}^{t-1}\partial f(\mathbf{X}_j)\mathbf{a}_{t-1-j,i}\Big\|^2}
	+ 3\,{\mathbb{E}\Big\|\sum_{j=0}^{t-1}\mathbf{V}_j\mathbf{a}_{t-1-j,i}\Big\|^2} \notag\\
		&=T_{3,i}+ T_{4,i}+ T_{5,i}.
		\label{eq:Qt-split}
	\end{align}
	Here, $\partial f(\mathbf{X}_j)\triangleq [\,\nabla f_1(\mathbf{x}_1^j)\ \cdots\ \nabla f_n(\mathbf{x}_n^j)\,]$.
	The three terms of \eqref{eq:Qt-split}, i.e., $T_{3,i}$ (stochastic noise), $T_{4,i}$ (bias due to disagreement), and $T_{5,i}$ (event-trigger perturbation), can be bounded using the spectral contraction
	$\|\mathbf{W}^k-\mathbf{J}\|_2^2\le \delta^k$, with $\delta=\|\mathbf{W}-\mathbf{J}\|_2^2\in[0,1)$.

	{{\textbf{Step 6}: Bounding $T_{3,i}$ (stochastic noise).}
	Recall
	\[
	T_{3,i}=\mathbb{E}\Big\|\sum_{j=0}^{t-1}\big(\mathbf{G}_j-\partial f(\mathbf{X}_j)\big)
	\mathbf{a}_{t-1-j,i}\Big\|^2.
	\]
	Let $\mathbf{a}_{t,i}\!\triangleq\!\frac{\mathbf{1}}{n}-\mathbf{W}^{k}\mathbf{e}_i$.
	Using $\|\mathbf{U}\mathbf{A}\|_F\!\le\!\|\mathbf{U}\|_F\|\mathbf{A}\|_2$ and
	$\|\mathbf{a}_{t,i}\|_2^2\!\le\!\delta^{k}$  we obtain
	\begin{align}
		T_{3,i}
		&\le \sum_{j=0}^{t-1}\mathbb{E}\big\|\mathbf{G}_j-\partial f(\mathbf{X}_j)\big\|_F^2\,\|\mathbf{a}_{t-1-j,i}\|_2^2\notag\\&
		\overset{(a)}{\le} \sum_{j=0}^{t-1} n\alpha^2\,\delta^{\,t-1-j}
		\le \frac{n\alpha^2}{1-\delta}.
		\label{eq:T3i}
	\end{align}
	where ($a$) is based on   Lemma \ref{lem:spectral-contraction}  proved in Appendix \ref{appb}.


	{{\textbf{Step 7}: Bounding $T_{4,i}$ (bias due to disagreement).}

\begin{align}
	&T_{4,i}
	= \mathbb{E}\bigg\|
	\sum_{j=0}^{t-1} \partial f(\mathbf{X}_j)\,\mathbf{a}_{t-1-j,i}
	\bigg\|_F^2
	\nonumber\\
	&\le
	\sum_{j=0}^{t-1} \mathbb{E}\big\|
	\partial f(\mathbf{X}_j)\,\mathbf{a}_{t-1-j,i}
	\big\|_F^2
	\nonumber\\
	&+
	\sum_{\substack{ j\neq j'}}^{t-1}
	\mathbb{E}\big\langle
	\partial f(\mathbf{X}_j)\,\mathbf{a}_{t-1-j,i},\,
	\partial f(\mathbf{X}_{j'})\,\mathbf{a}_{t-1-j',i}
	\big\rangle
	\nonumber\\
	&= \widetilde T_4 \;+\; \widetilde T_5.\label{T4_i}
\end{align}
Next, we proceed to bound $\widetilde T_4$ and $ \widetilde T_5$.
\begin{align}
	&\widetilde T_4
	=\sum_{j=0}^{t-1}\mathbb{E}\Big\|
	\partial f(\mathbf{X}_j)\,\mathbf{a}_{t-1-j,i}
	\Big\|^2 \nonumber\\
	&\le \sum_{j=0}^{t-1}\mathbb{E}\|\partial f(\mathbf{X}_j)\|_F^2\,
	\|\mathbf{a}_{t-1-j,i}\|_2^2 \nonumber
\end{align}
For $\mathbb{E}\|\partial f(\mathbf{X}_j)\|_F^2$, we have
\begin{align}
	&\mathbb{E}\|\partial f(\mathbf{X}_j)\|_F^2
	\le 3\mathbb{E}\Big\|
	\partial f(\mathbf{X}_j) - \partial f(\bar{\mathbf{x}}_j\mathbf{1}^\top)
	\Big\|_F^2
	\notag\\&+ 3\,\mathbb{E}\Big\|
	\partial f(\bar{\mathbf{x}}_j\mathbf{1}^\top) - \nabla f(\bar{\mathbf{x}}_j)\mathbf{1}^\top
	\Big\|_F^2 +\; 3\,\mathbb{E}\Big\|
	\nabla f(\bar{\mathbf{x}}_j)\mathbf{1}^\top
	\Big\|_F^2 \nonumber\\
	&\overset{(a)}{\le}\;
	3\,\mathbb{E}\Big\|
	\partial f(\mathbf{X}_j) - \partial f(\bar{\mathbf{x}}_j\mathbf{1}^\top)
	\Big\|_F^2
	+ 3n\beta^2
	+ 3\,\mathbb{E}\Big\|
	\nabla f(\bar{\mathbf{x}}_j)\mathbf{1}^\top
	\Big\|_F^2 \nonumber\\
	&\overset{(b)}{\le}\;
	3\sum_{i=1}^{n} L^2\,Q_{j,i}
	+ 3n\beta^2
	+ 3\,\mathbb{E}\Big\|
	\nabla f(\bar{\mathbf{x}}_j)\mathbf{1}^\top
	\Big\|_F^2 .
	\label{eq:bound-gradient-second-moment}
\end{align}
where ($a$) is due to Assumption \ref{ass:variance}, and ($b$) is due to the $L$-smoothness.  Hence,
\begin{align}
	&\widetilde T_4\le 3\sum_{j=0}^{t-1}\sum_{h=1}^{n}L^2\,\mathbb{E}Q_{j,h}\,
	\|\mathbf{a}_{t-1-j,i}\|_2^2 \nonumber\\
	&\quad+\;
	3\sum_{j=0}^{t-1}\mathbb{E}\Big\|
	\nabla f(\bar{\mathbf{x}}_j)\mathbf{1}^{\!\top}
	\Big\|_F^2
	\|\mathbf{a}_{t-1-j,i}\|_2^2 \nonumber\\
	&\quad+\;
	3n\beta^2\sum_{j=0}^{t-1}\|\mathbf{a}_{t-1-j,i}\|_2^2 .
	\label{eq:T4i-T4}
\end{align}
For $	\widetilde T_5$, we have
\begin{align}
	&\widetilde T_5
	=\sum_{j\neq j'}^{t-1}
	\mathbb{E}\Big\langle
	\partial f(\mathbf{X}_j)\,\mathbf{a}_{t-1-j,i},\,
	\partial f(\mathbf{X}_{j'})\,\mathbf{a}_{t-1-j',i}
	\Big\rangle \nonumber\\
	&\le \sum_{j\neq j'}^{t-1}
	\mathbb{E}\Big\|
	\partial f(\mathbf{X}_j)\,\mathbf{a}_{t-1-j,i}
	\Big\|\,
	\Big\|
	\partial f(\mathbf{X}_{j'})\,\mathbf{a}_{t-1-j',i}
	\Big\| \nonumber\\
	&\le \sum_{j\neq j'}^{t-1}
	\mathbb{E}\|\partial f(\mathbf{X}_j)\|\,
	\|\mathbf{a}_{t-1-j,i}\|_2\,
	\|\partial f(\mathbf{X}_{j'})\|\,
	\|\mathbf{a}_{t-1-j',i}\|_2 \nonumber\\
	&\le \sum_{j\neq j'}^{t-1}
	\frac{\mathbb{E}\|\partial f(\mathbf{X}_j)\|^2}{2}\,
	\|\mathbf{a}_{t-1-j,i}\|_2\,
	\|\mathbf{a}_{t-1-j',i}\|_2 \nonumber\\
	&\quad+\sum_{j\neq j'}^{t-1}
	\frac{\mathbb{E}\|\partial f(\mathbf{X}_{j'})\|^2}{2}\,
	\|\mathbf{a}_{t-1-j,i}\|_2\,
	\|\mathbf{a}_{t-1-j',i}\|_2 \nonumber\\
	&\overset{(a)}{\le} \sum_{j\neq j'}^{t-1}
	\mathbb{E}\Big(\frac{\|\partial f(\mathbf{X}_j)\|^2}{2}
	+\frac{\|\partial f(\mathbf{X}_{j'})\|^2}{2}\Big)\,
	\delta^{\,t-1-\frac{j+j'}{2}} \nonumber\\
	&= \sum_{j\neq j'}^{t-1}
	\mathbb{E}\|\partial f(\mathbf{X}_j)\|^2\,
	\delta^{\,t-1-\frac{j+j'}{2}} \nonumber\\
	&\overset{(b)}{\le} 3\sum_{j\neq j'}^{t-1}
	\Bigg(\sum_{i=1}^{n}L^2\,Q_{j,i}
	+\mathbb{E}\Big\|
	\nabla f(\bar{\mathbf{x}}_j)\mathbf{1}^\top
	\Big\|_F^2\Bigg)\,
	\delta^{\,t-1-\frac{j+j'}{2}} \nonumber\\
	&\quad+\;3n\beta^2\sum_{j\neq j'}^{t-1}
	\delta^{\,t-1-\frac{j+j'}{2}}\nonumber\\
	&= 	6\sum_{j=0}^{t-1}
	\Bigg(\sum_{h=1}^{n}\!L^2\,\mathbb{E}Q_{j,h}
	\!+\!\mathbb{E}\Big\|
	\nabla f(\bar{\mathbf{x}}_j)\mathbf{1}^{\!\top}
	\Big\|_F^2\Bigg)\!\!\!\!
	\sum_{j'=j+1}^{t-1}\!\!\!\!\!\sqrt{\delta}^{2t-j-j'-2}\nonumber
	\\&\quad+6n\beta^2 \sum_{j'>j}^{\,t-1} \delta^{\,t-1-\frac{j+j'}{2}}\nonumber\\
	&\leq 	6\sum_{j=0}^{t-1}
	\Bigg(\sum_{h=1}^{n}L^2\,\mathbb{E}Q_{j,h}
	+\mathbb{E}\Big\|
	\nabla f(\bar{\mathbf{x}}_j)\mathbf{1}^{\!\top}
	\Big\|_F^2\Bigg)
	\frac{\sqrt{\delta}^{t-j-1}}{1-\sqrt{\delta}}\nonumber
	\\&\quad+\frac{6n\beta^2}{(1-\sqrt{\delta})^2},\label{T_5wide}
\end{align}
where ($a$) is from Lemma \ref{lem:spectral-contraction}, and ($b$) stems from \eqref{eq:bound-gradient-second-moment}.
Plugging \eqref{eq:T4i-T4} and \eqref{T_5wide} into \eqref{T4_i} and using Lemma \ref{lem:spectral-contraction}, we have
\begin{align}
	&T_{4,i}
	\le \!\Bigg(3\sum_{j=0}^{t-1}\sum_{h=1}^{n}\E[L^{2}Q_{j,h}]
	+ \!3\!\sum_{j=0}^{t-1}\E\!\big\|\nabla f(\bar{\mathbf{x}}^{j})\mathbf{1}^{\top}\big\|_F^{2}\!\Bigg)\delta^{t-1-j}
	\nonumber\\
	& + \frac{6	(\sqrt{\delta})^{\,t-1-j}}{1-\sqrt{\delta}}
	\sum_{j=0}^{t-1}\!\Bigg(\sum_{h=1}^{n}\E[L^{2}Q_{j,h}]
	+ \E\!\big\|\nabla f(\bar{\mathbf{x}}^{\,j})\mathbf{1}^{\top}\big\|_F^{2}\Bigg)
	\nonumber\\
	& + \frac{9n\beta^{2}}{(1-\sqrt{\delta})^{2}} .
	\label{eq:T4i-final-converted}
\end{align}

	{{\textbf{Step 8}: Bounding $T_{5,i}$ (event-trigger perturbation).}
	With $\mathbf{a}_{t,i}$ as above and the Cauchy–Schwarz  inequality, it follows that
	\begin{align}
		&T_{5,i}
		=\mathbb{E}\Big\|\sum_{j=0}^{t-1}\mathbf{V}_j\mathbf{a}_{t-1-j,i}\Big\|^2
		\\&= \sum_{j=0}^{t-1}\mathbb{E}\|\mathbf{V}_j\mathbf{a}_{t-1-j,i}\|^2
		\!+\! \sum_{j\neq j'}^{t-1}\!\mathbb{E}\big\langle \mathbf{V}_j\mathbf{a}_{t-1-j,i},\mathbf{V}_{j'}\mathbf{a}_{t-1-j',i}\big\rangle \nonumber\\
		&\overset{(a)}{\le} \sum_{j=0}^{t-1}\mathbb{E}\|\mathbf{V}_j\|_F^2\delta^{t-1-j}
		\!+ \!\!\sum_{j\neq j'}^{t-1}\!\frac{1}{2}\mathbb{E}\big(\|\mathbf{V}_j\|_F^2+\|\mathbf{V}_{j'}\|_F^2\big)\delta^{t-\frac{j+j'}{2}-1} \nonumber\\
		&\le \sum_{j=0}^{t-1}\mathbb{E}\|\mathbf{V}_j\|_F^2\delta^{t-1-j}
		+ 2\sum_{j=0}^{t-1}\mathbb{E}\|\mathbf{V}_j\|_F^2\sum_{j'=j+1}^{t-1}\delta^{t-\frac{j+j'}{2}-1} \nonumber\\
		&\le \sum_{j=0}^{t-1}\mathbb{E}\|\mathbf{V}_j\|_F^2\left(\delta^{t-1-j}
		+ \frac{2{(\sqrt{\delta})}^{t-1-j}}{1-\sqrt{\delta}}\right),
		\label{eq:T5i}
	\end{align}
	where ($a$) arises from Lemma 1.

	{{\textbf{Step 9}: Bounding $Q_{t,i}$.}
	Plugging the bounds of $T_{3,i}$, $T_{4,i}$, and $T_{5,i}$ into
	\eqref{eq:Qt-split} (recall $Q_{t,i}\le 3\eta^2 T_{3,i}+3\eta^2 T_{4,i}+3T_{5,i}$) yields
	\begin{align}
		Q_{t,i}&
		\le \frac{3\eta^2 n\,\alpha^2}{1-\delta}
		\;+\; \frac{27\eta^2 n\,\beta^2}{(1-\sqrt{\delta})^2}+\left( \delta^{\,t-1-j} + \frac{2\,{(\sqrt{\delta})}^{\,t-1-j}}{1-\sqrt{\delta}} \right) \nonumber\\
		&\quad \times 9\eta^2 \sum_{j=0}^{t-1}\!
		\Big( L^2 \sum_{h=1}^{n}\mathbb{E}Q_{j,h}
		+ \mathbb{E}\big\|\nabla f(\bar{\mathbf{x}}^j)\mathbf{1}^\top\big\|_F^2 \Big)
		\! \nonumber\\
		&\quad + 3 \sum_{j=0}^{t-1}\mathbb{E}\|\mathbf{V}_j\|_F^2
		\left( \delta^{\,t-1-j} + \frac{2\,{(\sqrt{\delta})}^{\,t-1-j}}{1-\sqrt{\delta}} \right).
		\label{eq:Qti-final-delta}
	\end{align}
	
	{{\textbf{Step 10}: Averaging over devices.}
	Define the averaged disagreement $M_t\triangleq \frac{1}{n}\sum_{i=1}^n Q_{t,i}$.
	Using $\sum_{h=1}^{n}\mathbb{E}Q_{j,h}=n\,\mathbb{E}M_j$ and averaging
	\eqref{eq:Qti-final-delta} over $i$ gives
	\begin{align}
			&\mathbb{E}M_t
	\le \frac{3\eta^2 n\,\alpha^2}{1-\delta}
		\;+\; \frac{27\eta^2 n\,\beta^2}{(1-\sqrt{\delta})^2} \nonumber\\
		&\quad + 9\eta^2 \sum_{j=0}^{t-1}
		\mathbb{E}\big\|\nabla f(\bar{\mathbf{x}}^j)\mathbf{1}^\top\big\|_F^2
		\left( \delta^{\,t-1-j} + \frac{2\,{(\sqrt{\delta})}^{\,t-1-j}}{1-\sqrt{\delta}} \right) \nonumber\\
		&\quad + 9 n \eta^2 L^2 \sum_{j=0}^{t-1}\mathbb{E}M_j
		\left( \delta^{\,t-1-j} + \frac{2\,{(\sqrt{\delta})}^{\,t-1-j}}{1-\sqrt{\delta}} \right) \nonumber\\
		&\quad + 3 \sum_{j=0}^{t-1}\mathbb{E}\|\mathbf{V}_j\|_F^2
		\left( \delta^{\,t-1-j} + \frac{2\,{(\sqrt{\delta})}^{\,t-1-j}}{1-\sqrt{\delta}} \right).
		\label{eq:Mt-recursion-delta}
	\end{align}

{{\textbf{Step 11}: Bounding $T_2$ by the average disagreement.}
Recall $T_2=\mathbb{E}\big\|\nabla f(\bar{\mathbf{x}}_t)-\frac{1}{n}\sum_{i=1}^n \nabla f_i(\mathbf{x}_{i,t})\big\|^2$.
By the $L$–smoothness, we have
\begin{equation}
	\mathbb{E}\,T_2 \;\le\; L^2\,\mathbb{E}M_t,
	\qquad
	M_t \triangleq \frac{1}{n}\sum_{i=1}^n \mathbb{E}\|\mathbf{x}_{i,t}-\bar{\mathbf{x}}_t\|^2.
	\label{eq:T2-by-Mt}
\end{equation}

{{\textbf{Step 12}: Iterative inequality and telescoping.}
From \eqref{eq:iterate-ineq},
\begin{IEEEeqnarray}{rCl}
	\mathbb{E}\,f(\bar{\mathbf{x}}_{t+1})
	&\le& \mathbb{E}\,f(\bar{\mathbf{x}}_{t})
	\;-\; \frac{\eta}{2}\,\mathbb{E}\!\left\|\nabla f(\bar{\mathbf{x}}_{t})\right\|^2
	\;+\; \eta L^2\,\mathbb{E}M_t
	\nonumber\\[0.2em]
	&&{}+\; \frac{1}{\eta}\,\mathbb{E}\!\left\|\bar{\mathbf{e}}_{\,t}\right\|^2
	\;+\; \frac{\eta\,\alpha^2}{n}. \label{eq:ieee-step13}
\end{IEEEeqnarray}
Summing \eqref{eq:ieee-step13} for $t=0,\dots,T-1$ gives
\begin{IEEEeqnarray}{rCl}
	\frac{\eta}{2}\sum_{t=0}^{T-1}\mathbb{E}\!\left\|\nabla f(\bar{\mathbf{x}}_{t})\right\|^2
	&\le& f(\bar{\mathbf{x}}^{0})-f^{\ast}
	\;+\; \eta L^2 {\sum_{t=0}^{T-1}\mathbb{E}M_t}
	\nonumber\\[0.2em]
	&&{}+\; \frac{1}{\eta}\sum_{t=0}^{T-1}\mathbb{E}\!\left\|\bar{\mathbf{e}}_{\,t}\right\|^2
	\;+\; \frac{\eta\,\alpha^2 T}{n}.
	\label{eq:ieee-step13-sum}
\end{IEEEeqnarray}

{{\textbf{Step 13}: Bounding $\sum_{t=0}^{T-1}\mathbb{E}M_t$.}
From the recursion \eqref{eq:Mt-recursion-delta},
using $\frac{1}{1-\delta}\le \frac{1}{(1-\sqrt{\delta})^{2}}$, $\sum_{k=0}^{\infty}\delta^{k}=\frac{1}{1-\delta}$ and
$\sum_{k=0}^{\infty}(\sqrt{\delta})^{k}=\frac{1}{1-\sqrt{\delta}}$, we have
\begin{align}
	&\sum_{t=0}^{T-1}\mathbb{E}M_t
	\le \left(\frac{3\eta^2 n\alpha^2}{1-\delta}
	+ \frac{27\eta^2 n\beta^2}{(1-\sqrt{\delta})^{2}}\right) T
	\nonumber\\[0.2em]
	&{}+\; \frac{27\eta^2}{(1-\sqrt{\delta})^{2}}
	\sum_{t=0}^{T-1}\!\mathbb{E}\!\left\|
	\nabla f(\bar{\mathbf{x}}_{t})\mathbf{1}^{\top}\right\|_F^{2}
	\;+\; \frac{27n\eta^2 L^2}{(1-\sqrt{\delta})^{2}}
	\sum_{t=0}^{T-1}\!\mathbb{E}M_t
	\nonumber\\[0.2em]
	&{}+\; \frac{9}{(1-\sqrt{\delta})^{2}}
	\sum_{t=0}^{T-1}\!\mathbb{E}\!\left\|\mathbf{V}_{t}\right\|_F^{2}.
	\label{eq:ieee-SM-final}
\end{align}

\noindent{Equivalently,}
\begin{align}
	&\Big(1 - \frac{27n\eta^2 L^2}{(1-\sqrt{\delta})^{2}}\Big)
	\sum_{t=0}^{T-1}\!\mathbb{E}M_t
	\le \left(\frac{3\eta^2 n\alpha^2}{1-\delta}
	+ \frac{27\eta^2 n\beta^2}{(1-\sqrt{\delta})^{2}}\right) T
	\nonumber\\[0.2em]
	&{}+\! \frac{27\eta^2}{(1-\sqrt{\delta})^{2}}
	\sum_{t=0}^{T-1}\!\mathbb{E}\!\left\|
	\nabla f(\bar{\mathbf{x}}_{t})\mathbf{1}^{\top}\right\|_F^{2}
	\!+\! \frac{9}{(1-\sqrt{\delta})^{2}}
	\sum_{t=0}^{T-1}\!\mathbb{E}\!\left\|\mathbf{V}_{t}\right\|_F^{2}.
	\label{eq:ieee-SM-isolate}
\end{align}

{{\textbf{Step 14}: Bounding $\sum_{t=0}^{T-1}\mathbb{E}M_t$.}
From \eqref{eq:ieee-SM-isolate},
define
\[
\Gamma \;\triangleq\; 1-\frac{27\,n\,\eta^2 L^2}{(1-\sqrt{\delta})^{2}}\;>\;0 .
\]
Using $\|\nabla f(\bar{\mathbf{x}}_{t})\mathbf{1}^\top\|_F^2
= n\,\|\nabla f(\bar{\mathbf{x}}_{t})\|^2$, \eqref{eq:ieee-SM-isolate} becomes
\begin{align}
	&\sum_{t=0}^{T-1}\mathbb{E}M_t
	\le \frac{1}{\Gamma}\Bigg[
	\left(\frac{3\eta^2 n\alpha^2}{1-\delta}
	+\frac{27\eta^2 n\beta^2}{(1-\sqrt{\delta})^{2}}\right) T
	\nonumber\\[0.2em]
	&
	+\ \frac{27 n \eta^2}{(1-\sqrt{\delta})^{2}}
	\sum_{t=0}^{T-1}\mathbb{E}\|\nabla f(\bar{\mathbf{x}}_{t})\|^{2}
	\!+\! \frac{9}{(1-\sqrt{\delta})^{2}}
	\sum_{t=0}^{T-1}\mathbb{E}\|\mathbf{V}_{t}\|_F^{2}
	\Bigg].
	\label{eq:ieee-SM-solved-compact}
\end{align}

{{\textbf{Step 15}: Final error bound.}
Plug \eqref{eq:ieee-SM-solved-compact} into \eqref{eq:ieee-step13-sum}:
\begin{align}
	&\frac{\eta}{2}\sum_{t=0}^{T-1}\mathbb{E}\|\nabla f(\bar{\mathbf{x}}_{t})\|^{2}
	\notag\\&\le f(\bar{\mathbf{x}}^{0})-f^{\ast}
	+ \frac{\eta L^{2}}{\Gamma}\!\left(\frac{3\eta^{2}n\alpha^{2}}{1-\delta}
	+ \frac{27\eta^{2}n\beta^{2}}{(1-\sqrt{\delta})^{2}}\right)T \nonumber\\
	& + \frac{27 n \eta^{3} L^{2}}{(1-\sqrt{\delta})^{2}\Gamma}
	\sum_{t=0}^{T-1}\mathbb{E}\|\nabla f(\bar{\mathbf{x}}_{t})\|^{2}
	\!+\! \frac{9 \eta L^{2}}{(1-\sqrt{\delta})^{2}\Gamma}
	\sum_{t=0}^{T-1}\mathbb{E}\|\mathbf{V}_{t}\|_{F}^{2}  \nonumber\\
	& + \frac{1}{\eta}\sum_{t=0}^{T-1}\mathbb{E}\|\bar{\mathbf{e}}_{t}\|^{2}
	+ \frac{\eta\alpha^{2}T}{n}.
	\label{eq:ieee-final-pre}
\end{align}

Move the gradient–sum  $\frac{27 n \eta^{3} L^{2}}{(1-\sqrt{\delta})^{2}\Gamma}
\sum_{t=0}^{T-1}\mathbb{E}\|\nabla f(\bar{\mathbf{x}}_{t})\|^{2}$ on the right  to the left of \eqref{eq:ieee-final-pre}; then we obtain
\begin{align}
	&\eta\!\left(\frac{1}{2}-\frac{27 n \eta^2 L^2}{(1-\sqrt{\delta})^{2}\,\Gamma}\right)
	\sum_{t=0}^{T-1}\mathbb{E}\|\nabla f(\bar{\mathbf{x}}_{t})\|^2
	\le f(\bar{\mathbf{x}}^{0})-f^{\ast}
	\;
	\nonumber\\
	&+ \frac{\eta L^2}{\Gamma}
	\left(\frac{3\eta^2 n\alpha^2}{1-\delta}
	+\frac{27\eta^2 n\beta^2}{(1-\sqrt{\delta})^{2}}\right) T+\; \frac{\eta\,\alpha^2 T}{n}\notag\\
	&{}+\; \frac{\eta L^2}{(1-\sqrt{\delta})^{2}\Gamma}
	\,9 \sum_{t=0}^{T-1}\mathbb{E}\|\mathbf{V}_{t}\|_F^{2}
	\;+\; \frac{1}{\eta}\sum_{t=0}^{T-1}\mathbb{E}\|\bar{\mathbf{e}}_{\,t}\|^2
	\;.
	\label{eq:ieee-final-correct}
\end{align}
Notice that 
\[
\|\mathbf V_t\|_F^2 = \sum_{i=1}^n \|\mathbf v_{i,t}\|^2 
\;\le\; \sum_{i=1}^n \tau_t^2 
= n\tau_t^2,
\]
and
\[
\|\bar{\mathbf e}_t\|
= \Big\|\frac{1}{n}\sum_{i=1}^n \mathbf v_{i,t}\Big\|
\;\le\;\frac{1}{n}\sum_{i=1}^n \|\mathbf v_{i,t}\|
\;\le\;\frac{1}{n}\sum_{i=1}^n \tau_t
= \tau_t,
\]
which yields $\|\bar{\mathbf e}_t\|^2 \le \tau_t^2$. 
Plugging these bounds into~(49), we obtain the final inequality:
\begin{align}
	&\eta\!\left(\frac{1}{2}-\frac{27n\eta^2L^2}{(1-\sqrt{\delta})^2\Gamma}\right)
	\sum_{t=0}^{T-1}\E\|\nabla f(\bar{\mathbf x}_t)\|^2
\le f(\bar{\mathbf x}^0)-f^\ast	\nonumber\\
&\quad
	+ \frac{\eta L^2}{\Gamma}\!\left(\frac{3n^2\alpha^2}{1-\delta}
	+ \frac{27\eta^2 n\beta^2}{(1-\sqrt{\delta})^{2}}\right)T
	+ \frac{\eta\alpha^2 T}{n}
	\nonumber\\
	&\quad
	+ \frac{9\eta L^2}{(1-\sqrt{\delta})^{2}\Gamma}
	\sum_{t=0}^{T-1} n\tau_t^2
	+ \frac{1}{\eta}\sum_{t=0}^{T-1}\tau_t^2.
	\label{eq:final-bound}
\end{align}

\section{}\label{appc}
\begin{lemma}\label{lem:spectral-contraction}
	Under Assumption~1, 
	 for any $i\in\{1,\ldots,n\}$ and any $t\in\mathbb{N}$,
	\[
	\big\|\mathbf{a}_{t,i}\big\|_2^2
	=\Big\|\frac{1}{n}\mathbf{1}-\mathbf{W}^t\mathbf{e}_i\Big\|_2^2
	\;\le\; \delta^{\,t}.
	\]
\end{lemma}

\begin{proof}
	Because $\mathbf{W}$ is real symmetric and doubly stochastic, it is diagonalizable by an orthonormal basis:
	\[
	\mathbf{W}=\mathbf{U}\,\mathrm{diag}(1,\lambda_2,\ldots,\lambda_n)\,\mathbf{U}^\top,
	\quad \text{with}\quad |\lambda_\ell|<1\ \ (\ell\ge2).
	\]
	Moreover, $\mathbf{J}=\frac{1}{n}\mathbf{1}\mathbf{1}^\top$ is the orthogonal projector onto $\mathrm{span}\{\mathbf{1}\}$, hence
	\[
	\mathbf{J}=\mathbf{U}\,\mathrm{diag}(1,0,\ldots,0)\,\mathbf{U}^\top,
	\mathbf{W}-\mathbf{J}
	=\mathbf{U}\,\mathrm{diag}(0,\lambda_2,\ldots,\lambda_n)\,\mathbf{U}^\top .
	\]
	Therefore,
	\[
	\mathbf{W}^t-\mathbf{J}
	=\mathbf{U}\,\mathrm{diag}(0,\lambda_2^{\,t},\ldots,\lambda_n^{\,t})\,\mathbf{U}^\top,
	\|\mathbf{W}^t-\mathbf{J}\|_2=\max_{\ell\ge2}|\lambda_\ell|^{\,t}.
	\]
	Since $\|\mathbf{W}-\mathbf{J}\|_2=\max_{\ell\ge2}|\lambda_\ell|$ and $\delta=\|\mathbf{W}-\mathbf{J}\|_2^2$, we obtain
	\[
	\|\mathbf{W}^t-\mathbf{J}\|_2 \;=\; \|\mathbf{W}-\mathbf{J}\|_2^{\,t} \;=\; (\sqrt{\delta})^{\,t}.
	\]
	Using $\mathbf{a}_{t,i}=(\mathbf{J}-\mathbf{W}^t)\mathbf{e}_i$ and $\|\mathbf{e}_i\|_2=1$,
\begin{align}
	\|\mathbf{a}_{t,i}\|_2
	=\|(\mathbf{J}-\mathbf{W}^t)\mathbf{e}_i\|_2
	\le \|\mathbf{J}-\mathbf{W}^t\|_2\,\|\mathbf{e}_i\|_2
	= (\sqrt{\delta})^{\,t}.\label{delta_sqrt}
    \end{align}
	Squaring both sides of \eqref{delta_sqrt} yields $\|\mathbf{a}_{t,i}\|_2^2\le \delta^{\,t}$, completing the proof.
\end{proof}

\end{document}